\documentclass[10pt, two column, twoside]{IEEEtran}
\usepackage{amssymb}
\usepackage{amsmath} 
\usepackage[lined,boxed,commentsnumbered, ruled]{algorithm2e}
\usepackage{dsfont}
\usepackage{makecell}
\usepackage{mathrsfs}
\usepackage{bm}
\usepackage{tikz}
\usepackage{caption}
\usepackage{graphicx,booktabs,multirow}
\usepackage{epstopdf}
\usepackage{subfigure}
\usepackage{diagbox}
\usepackage{multirow}
\usepackage{tabularx} 
\usepackage{booktabs}
\usepackage{float}
\usepackage{verbatim}
\usepackage{color}
\usepackage[colorlinks, linkcolor=black, anchorcolor=black, citecolor=black]{hyperref}
\usepackage{cite}
\usepackage{setspace}

\definecolor{colorhkust}{RGB}{20,43,140}
\definecolor{colortsinghua}{RGB}{116,52,129}
\definecolor{color1}{RGB}{128,0,0}
\usepackage{color}
\usepackage{amsthm}

\usepackage[T1]{fontenc}
\usepackage{amsmath}
\usepackage[cmintegrals]{newtxmath}
\usetikzlibrary{arrows}
\theoremstyle{definition} 
\newtheorem{lemma}{Lemma}
\newtheorem{theorem}{Theorem}

\newtheorem{proposition}{Proposition}
\newtheorem{definition}{Definition}

\newtheorem{assumption}{Assumption}

\begin{document}

\title{One-Bit Byzantine-Tolerant Distributed Learning via Over-the-Air Computation}
\author{Yuhan~Yang, \textit{Student Member}, \textit{IEEE}, Youlong Wu, \textit{Member}, \textit{IEEE}, Yuning Jiang, \textit{Member}, \textit{IEEE}, \\ and Yuanming~Shi, \textit{Senior Member}, \textit{IEEE}
	
\thanks{The work of Yuanming Shi was supported in part by the National Nature Science Foundation of China under Grant 62271318, the Natural Science Foundation of Shanghai under Grant No. 21ZR1442700, and the Shanghai Rising-Star Program under Grant No. 22QA1406100. The work of Yuning Jiang was supported by the Swiss National Science Foundation under the NCCR Automation (grant agreement 51NF40\_180545). \emph{(Corresponding author: Youlong Wu and Yuanming Shi.)}}
\thanks{Y. Yang, Y. Wu, and Y. Shi are with the School of Information Science and Technology, ShanghaiTech University, Shanghai 201210, China (e-mail: \{yangyh1, wuyl1, shiym\}@shanghaitech.edu.cn).}
\thanks{Y. Jiang is with the Automatic Control Laboratory, EPFL, Switzerland (yuning.jiang@epfl.ch).}
\thanks{This work has been submitted to the IEEE for possible publication. Copyright may be transferred without notice, after which this version may no longer be accessible.}
}


\maketitle

\begin{abstract}
Distributed learning has become a promising computational parallelism paradigm that enables a wide scope of intelligent applications from the Internet of Things (IoT) to autonomous driving and the healthcare industry. This paper studies distributed learning in wireless data center networks, which contain a central edge server and multiple edge workers to collaboratively train a shared global model and benefit from parallel computing. However, the distributed nature causes the vulnerability of the learning process to faults and adversarial attacks from Byzantine edge workers, as well as the severe communication and computation overhead induced by the periodical information exchange process. To achieve fast and reliable model aggregation in the presence of Byzantine attacks, we develop a signed stochastic gradient descent (SignSGD)-based Hierarchical Vote framework via over-the-air computation (AirComp), where one voting process is performed locally at the wireless edge by taking advantage of Bernoulli coding while the other is operated over-the-air at the central edge server by utilizing the waveform superposition property of the multiple-access channels. We comprehensively analyze the proposed framework on the impacts including Byzantine attacks and the wireless environment (channel fading and receiver noise), followed by characterizing the convergence behavior under non-convex settings. Simulation results validate our theoretical achievements and demonstrate the robustness of our proposed framework in the presence of Byzantine attacks and receiver noise.

\end{abstract}


\IEEEpeerreviewmaketitle

\section{Introduction} \label{SecI}
Fueled by the unprecedented success of artificial intelligence (AI), the explosion of abundant data generated from myriad edge devices with advanced sensing, communication, and computation technologies (e.g., smart phones, robots, and vehicles) can be utilized to boost a large number of emerging intelligence services and applications \cite{Distributed-learning-intro1}. Nevertheless, although the massive amount of data is capable of enhancing the performance of machine learning, due to the limited computational resources, it is always hard for a single entity to handle such high-volume, open-ended datasets effectively and efficiently. As a remedy, distributed learning \cite{Distributed-learning-intro2} has become a promising alternative which orchestrates multiple edge workers to collaboratively train a shared global model and enjoys computational parallelism. Celebrated distributed learning paradigms such as federated learning \cite{FLChallenge-future,FLGoogle,Distributed-learning-intro1,Distributed-learning-intro2}, swarm learning \cite{Swarm-learning}, and split learning \cite{Split-learning-intro}, have realized a wide scope of applications including 6G networks \cite{Distributed-learning-intro1,Distributed-learning-intro2,Distributed-learning-intro3,Distributed-learning-intro5}, Internet of Things (IoT) \cite{Iot-intro}, autonomous driving \cite{Distributed-learning-intro4}, and healthcare industry \cite{Healthcare-intro}. In this paper, we consider a distributed learning paradigm named data center networks \cite{Intro-data-center1,Intro-data-center2,Intro-data-center3} where a central edge server with access to the entire global dataset partitions the training samples into several non-overlapping sub-datasets and distributes the datasets to multiple edge workers for parallel computing, followed by orchestrating the edge workers to collaboratively train a shared global model by exchanging local update results (e.g., model parameters or gradients). However, the distributed nature also raises concerns about trustworthiness, especially for high-stake applications (e.g., autonomous driving), which require high guarantees in terms of privacy, security, and fairness during the learning process \cite{Distributed-learning-intro3,Distributed-learning-intro5,Trustworthy-AI-intro,Distributed-learning-intro6}. Besides, due to the periodical model update exchange process, the induced communication and computation overhead is another non-negligible factor that limits the performance of distribution learning \cite{Distributed-learning-intro3}.

\par
Although the distributed nature benefits the data center networks from computational parallelism, it causes severe security concerns due to its vulnerability to faults and adversarial attacks, such as Byzantine attacks \cite{FLChallenge-future}. In Byzantine settings, several edge workers become untrustworthy or even adversarial in both communication and computation, i.e., Byzantine edge workers, which can transmit malicious messages to the central edge server and mislead the learning process. As revealed in \cite{FLByzantine-SGD2,Coordinate-wise-trimmed-mean,Krum}, even one Byzantine edge worker can arbitrarily manipulate the global model and incur significant performance degradation, yielding Byzantine fault tolerance a critical consideration in distributed learning. To counter Byzantine attacks, multiple robust aggregation-based Byzantine-tolerant schemes have been proposed recently. In particular, instead of directly computing the naive mean, the central edge server estimates the global update by utilizing geometric median (GM) \cite{Geometric-median,FLByzantine-SGD3,FLRobust-GM1,FLRobust-GM2}, coordinate median \cite{Coordinate-wise-trimmed-mean}, trimmed median \cite{Coordinate-wise-trimmed-mean}, Krum \cite{Krum}, robust stochastic aggregation \cite{RSA}, and iterative filtering \cite{Iterative-filtering} to tolerate Byzantine attacks. Besides, the authors in \cite{SignSGD2} proposed Majority Vote-based SignSGD by leveraging one-bit gradient quantization to improve communication efficiency and defend against attacks from a small number of Byzantine edge workers, and \cite{Election-coding} further enhanced Byzantine tolerance by detecting redundant gradient computation, as in \cite{DETOX,DRACO}. Moreover, the authors in \cite{Mimic} demonstrated that combining the robust aggregation rules with a re-sampling strategy can achieve Byzantine robustness while alleviating the impact of data heterogeneity, which extends the research of Byzantine tolerance to not identically and independently distributed (non-i.i.d.) scenarios.

\par
Unfortunately, the huge computational cost induced by robust aggregation and the high communication overhead caused by periodic update exchange are intolerable in current distributed learning systems, which requires additional mechanisms for efficient model aggregation. Recently, by integrating communication and computation, over-the-air computation (AirComp) has become a leading analog model aggregation paradigm, which is achieved by exploiting the waveform superposition property of a multiple-access channel, yielding high spectrum efficiency and low transmission latency \cite{FLAirComp1,FLAirComp2,FLPrivacy-free}. Specifically, the authors in \cite{FLAirComp2} leveraged AirComp for fast and reliable model aggregation by jointly designing receiver beamforming and device scheduling strategy. A gradient sparsification and random linear projection scheme were specialized in \cite{Channel-inverse1} to further reduce the dimension of model updates, thereby shedding light on the deployment of AirComp in bandwidth-limited scenarios. Besides, the authors in \cite{FLPrivacy-free} exploited the inevitable receiver noise to alleviate the privacy concerns for free by developing an adaptive power control scheme in the AirComp-enabled distributed learning systems. To further enhance communication efficiency and facilitate the deployment of AirComp in the current wireless communication systems, the authors in \cite{One-bit-AirComp1} introduced digital modulation strategies into AirComp, which is achieved by leveraging one-bit quantization of gradient with digital quadrature amplitude modulation (QAM) at the edge workers, and over-the-air Majority Vote at the central edge server.

\par
Despite the communication and computation benefits, the naive mean operation of AirComp makes it vulnerable to Byzantine attacks, yielding severe security concerns \cite{Byzantine-AirComp1}. To address this issue, the authors in \cite{Byzantine-AirComp1} approximated the GM by performing smoothed Weiszfeld algorithm over-the-air in an alternative manner, and the authors in \cite{Byzantine-AirComp2} proposed ROTAF, where the participated edge workers are divided into several disjoint clusters and the central edge server aggregates the update messages of each cluster via AirComp, followed by performing GM to achieve robust aggregation. The authors in \cite{Byzantine-AirComp3} developed a best-effort voting-based power control policy by enabling the edge workers to transmit with the maximum power to defend against Byzantine attacks. However, the current AirComp-enabled Byzantine-tolerant mechanisms still encounter several hindrances which should be taken into consideration. In particular, the required multiple aggregations in each learning round to achieve robust aggregation entails high communication and computation cost \cite{Byzantine-AirComp1,Byzantine-AirComp2}. Besides, as revealed in \cite{FLByzantine-SGD2}, with the increasing dimension of model parameters, the computation overhead required for the robust estimator, e.g., GM, is even higher than that of mini-batch gradient computation, which is unaffordable in the current distributed learning systems. Further, the existence of unfavorable channel propagations including fading and receiver noise degenerates the performance of robust aggregation, thereby restricting the tolerable number of Byzantine edge workers into a small range \cite{Byzantine-AirComp1,Byzantine-AirComp2,Byzantine-AirComp3}. Therefore, it is significant to design additional mechanisms to enhance communication efficiency while maintaining Byzantine tolerance in the distributed learning systems.

\subsection{Contributions}
Motivated by the aforementioned problems, we utilize the accessibility of data center networks to the entire global dataset, combined with AirComp, to design a Byzantine-tolerant and communication-efficient framework. Specifically, in this paper, instead of concentrating on the robust aggregation rules, we propose to leverage the redundant computation strategy \cite{DETOX,DRACO,Election-coding} for Byzantine robustness, resulting in a two-layer Hierarchical Vote framework. Besides, inspired by the hardware-friendly property and the robustness to noise \cite{One-bit-AirComp1} and Byzantine attacks \cite{SignSGD2} of the Majority Vote-based SignSGD, we develop an AirComp-enabled Hierarchical Vote framework in wireless data center networks for fast and reliable model aggregation, which paves the road to achieve communication-efficient and Byzantine-tolerant distributed learning systems.

\par
The major contributions of this article are summarized as follows: 1) From a systematic perspective, we develop a Hierarchical Vote framework in wireless distributed data center networks to cope with Byzantine attacks, followed by leveraging AirComp for fast model aggregation. 2) From a theoretical perspective, we demonstrate the robustness of the developed scheme in the presence of Byzantine attacks and receiver noise by characterizing the global decoding bit error probability and the convergence behavior. 3) Numerical simulations are conducted to validate the theoretical achievements and demonstrate that the proposed framework can achieve robustness to Byzantine attacks and receiver noise. Besides, by comparing with the existing mechanisms, we demonstrate the advancements of the proposed framework in reducing communication and computation cost.

\subsection{Organization and Notations}
The remainder of this article is organized as follows: Section \ref{SecII} introduces the AirComp-enabled Hierarchical Vote framework in the wireless data center networks under Byzantine settings. In Section \ref{SecIII}, we theoretically analyze the Byzantine-tolerant error bound of the proposed framework, followed by characterizing the convergence behavior under non-convex scenarios. Simulation results are presented in Section \ref{SecIV} to demonstrate the advantages of the AirComp-enabled Hierarchical Vote. Finally, Section \ref{SecV} concludes this work.

\par
\emph{Notations}: Italic and boldface letters denote scalar and vector (matrix), respectively. $\mathbb{R}^{m\times n}$ and $\mathbb{C}^{m\times n}$ denote the real and complex domains with the space of $m\times n$, respectively. For a positive integer $i$, we let $[i]\triangleq\{1,\ldots,i\}$. The operators $(\cdot)^T,(\cdot)^H$, and $\mathbb{E}\,(\cdot)$ represent the transpose, Hermitian transpose, and statistical expectation, respectively. $\pmb I_n$ denotes the identity matrix with the space of $n\times n$. $\pmb E_{ij}$ denotes the element in row $i$ and column $j$ in matrix $\pmb E$. The operator $|\cdot|$ is the cardinality of a set or the absolute value of a scalar number, and $||\cdot||$ denotes the Euclidean norm.

\section{System Model}\label{SecII}
In this section, we first elaborate on a one-bit distributed learning paradigm in the presence of Byzantine attacks, followed by proposing a robust Hierarchical Vote scheme to suppress the undesired bias during the learning process. Furthermore, over-the-air computation is introduced to support fast model aggregation.

\subsection{One-Bit Distributed Learning Protocol} \label{SecII-A}
As shown in Fig. \ref{fig:Distributed-learning-model}, we consider a wireless distributed data center network \cite{Intro-data-center1,Intro-data-center2,Intro-data-center3} consisting of one single-antenna central edge server, and $K$ single-antenna edge workers, which are indexed by set $\mathcal{K}=[K]$. Suppose that the central edge server has access to a global dataset $\mathcal{D}$ containing $N$ training data points, which are i.i.d. samples from a common distribution \cite{FLPrivacy-free,FLRobust-GM1,Intro-data-center3}. Specifically, in this paper, we mainly focus on the supervised learning applications \cite{FLRobust-GM1} where the training data point $\pmb\xi=(\pmb x,y)\in\mathcal{D}$ is composed of a feature vector $\pmb x$ and its corresponding label $y$. In the wireless data center network, the central edge server uniformly partitions the global dataset $\mathcal{D}$ into $K$ non-overlapping sub-datasets $\{\mathcal{D}_i\}_{i=1}^{K}$ with the same size $D=\lfloor N/K\rfloor$, i.e., $\mathcal{D}_i \cap\mathcal{D}_j=\emptyset$, $\cup_{i=1}^N\mathcal{D}_i=\mathcal{D}$, followed by allocating them to each edge worker based on a preset data allocation matrix $\pmb E\in\{0,1\}^{K\times K}$. The data allocation strategy states that sub-dataset $\mathcal{D}_i$ is allocated to the $k$-th edge worker if $\pmb E_{ki}=1$, and $\pmb E_{ki}=0$ otherwise. For instance, $\pmb E=\pmb I_K$ represents the canonical distributed learning settings where the dataset $\mathcal{D}_i$ is assigned to the $i$-th edge worker \cite{SignSGD2}. We define $\mathcal{S}_k=\{i:\pmb E_{ki}=1\}$ as the index set of datasets allocated to the $k$-th edge worker. Besides, we assume that the data allocation process is completely offline, i.e., all edge workers can receive the allocated datasets from the central edge server without any distortion \cite{Intro-data-center2}.

\par
For a given $d$-dimensional global model parameter $\pmb\omega\in\mathbb{R}^d$, denote $F(\pmb\omega)$ as the empirical global loss function. The goal of this data center network is to solve the following empirical risk minimization (ERM) problem in a distributed fashion:
\begin{equation}
\underset{\pmb\omega\in\mathbb{R}^d}{\text{minimize}}\;\;F(\pmb\omega):=\frac{1}{N}\sum\nolimits_{\pmb\xi\in\mathcal{D}}f(\pmb\omega,\pmb\xi), \label{Global-loss-function}
\end{equation}
where $f(\pmb\omega,\pmb\xi)$ denotes the sample-wise loss function measuring the training error of model $\pmb\omega$ on the data point $\pmb\xi\in\mathcal{D}$. An effective method to address problem (\ref{Global-loss-function}) is distributed mini-batch SGD, where the stochastic gradients from all edge workers are aggregated and averaged at the central edge server to update the global model $\pmb\omega$ iteratively.

\begin{figure}[t]
	\centering
	\includegraphics[width=0.45\textwidth]{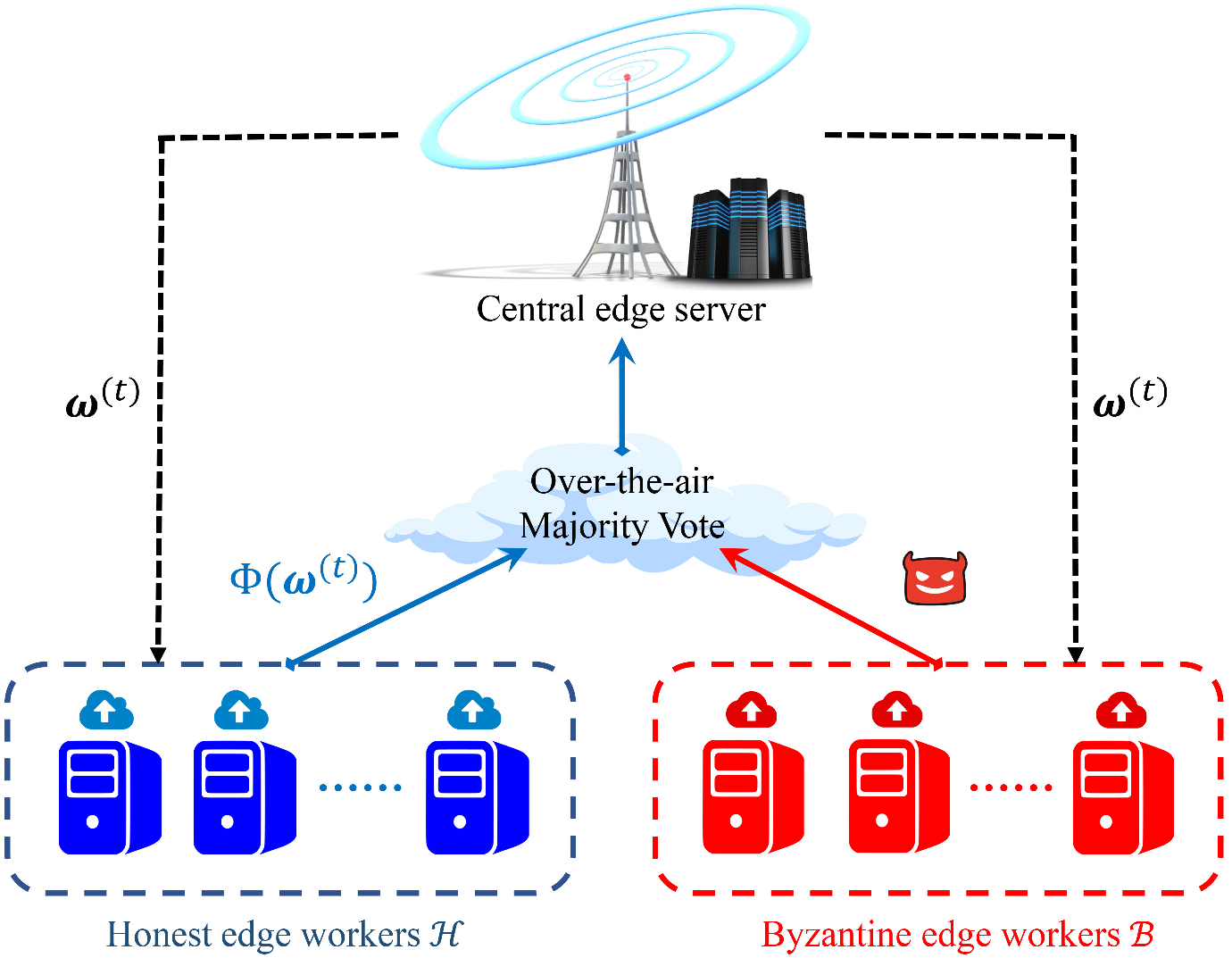}
	\caption{A wireless distributed data center network in the presence of Byzantine edge workers.}
	\label{fig:Distributed-learning-model}
\end{figure}

\par
To reduce the communication cost in the wireless distributed data center network and maintain the convergence rate of SGD, Majority Vote-based SignSGD \cite{Election-coding,SignSGD1,SignSGD2,One-bit-AirComp1} has become a prevalent distributed learning paradigm, where each edge worker encodes its transmitted stochastic gradient into a binary vector and the aggregated gradient is decoded via Majority Vote at the central edge server. Specifically, in the $t$-th learning round, the central edge server broadcasts the current global model parameter $\pmb\omega^{(t)}$ to all edge workers. Then, each edge worker $k\in\mathcal{K}$ independently selects a random mini-batch $\mathcal{A}_k^{(t)}\subseteq\mathcal{D}_k$ with cardinality $A$ \cite{Channel-inverse1}, followed by computing its local stochastic gradient $\pmb g_k^{(t)}$ based on $\pmb\omega^{(t)}$, which is given by:
\begin{equation}
\pmb g_k^{(t)}=\frac{1}{A}\sum\nolimits_{\pmb\xi\in\mathcal{A}_k^{(t)}}\nabla f(\pmb\omega^{(t)},\pmb\xi), \label{Stochastic-gradient}
\end{equation}
and the corresponding quantified one-bit stochastic gradient $\pmb m_k^{(t)}\in\{-1,1\}^{d}$ can be obtained by extracting the sign of each element in the local gradient $\pmb g_k^{(t)}$, i.e.,
\begin{equation}
\pmb m_k^{(t)}(j)=\text{sign}\big[\pmb g_k^{(t)}(j)\big],\;\;\forall j\in[d], \label{Quantified-gradient}
\end{equation}
where $\pmb m_k^{(t)}(j)$ denotes the $j$-th entry of $\pmb m_k^{(t)}$ and so as $\pmb g_k^{(t)}(j)$. The above local update process is operated in parallel across $K$ edge workers, and all edge workers upload the quantified gradients $\{\pmb m_k^{(t)}\}_{k=1}^{K}$, yielding a one-bit global update $\pmb v^{(t)}$ at the edge server. In particular, $\pmb v^{(t)}$ is a one-bit estimation of the global model update $\pmb g^{(t)}=\nabla F(\pmb\omega^{(t)})$ and can be obtained via element-wise Majority Vote based on the aggregated $\{\pmb m_k^{(t)}\}_{k=1}^{K}$ at the edge server \cite{SignSGD2}, which is given by
\begin{equation}
\pmb v^{(t)}(j)=\text{sign}\Big[\sum_{k=1}^{K}\pmb m_k^{(t)}(j)\Big]\in\{-1,1\},\;\;\forall j\in[d], \label{Majority-vote}
\end{equation}
where $\pmb v^{(t)}=[\pmb v^{(t)}(1),\ldots,\pmb v^{(t)}(d)]^T\in\{-1,1\}^d$ denotes the global model update. In essence, the Majority Vote process in (\ref{Majority-vote}) establishes the global model update by taking the majority opinions among the results from $K$ edge workers, which is achieved by element-wise selecting the more frequent element ($-1$ or $1$) in $\{\pmb m_k^{(t)}\}_{k=1}^{K}$. Finally, the central edge server updates the global model via gradient descent, i.e., $\pmb\omega^{(t+1)}=\pmb\omega^{(t)}-\eta\pmb v^{(t)}$
where $\eta>0$ denotes the learning rate. The global model $\pmb\omega$ is periodically updated until the given convergence criterion, e.g., a maximum number of learning rounds $T$ is reached.

\subsection{Byzantine Attack and Byzantine Tolerance}
The distributed nature empowers the wireless data center network as a promising paradigm for computational parallelism, however, it also causes the vulnerability of this system to faults and adversarial attacks. As revealed in \cite{SignSGD2}, the naive Majority Vote-based SignSGD, i.e., $\pmb E=\pmb I_K$, only achieves robustness when the number of adversarial edge workers is restricted in a small range, which requires additional mechanisms to provide stronger security guarantees. To support further analysis, we present the Byzantine settings and the definition of Byzantine tolerance in this section.

\par
Fig. \ref{fig:Distributed-learning-model} illustrates a wireless distributed data center network in the presence of Byzantine attacks, where there exist $B$ Byzantine edge workers that can upload malicious model updates to mislead the learning process. For brevity, we denote $\mathcal{B}$ as the set of Byzantine edge workers and $c=B/K$ as the corruption level to represent the fraction of Byzantine edge workers. Besides, we assume a practical scenario where the identities of the Byzantine edge workers are unavailable to the central edge server and the other honest edge workers during the learning process \cite{Coordinate-wise-trimmed-mean}. We denote the corrupted updates sent by a Byzantine edge worker as an arbitrary $d$-dimensional vector $\pmb *$, which can be strategically biased by the Byzantine edge workers to inject corrupted information to the system \cite{FLRobust-GM1,FLByzantine-SGD3}. Hence, to include the Byzantine settings, the local update oracle can be modified as follows \cite{Krum,Coordinate-wise-trimmed-mean,FLRobust-GM1}:
\begin{equation}
\pmb m_k^{(t)}=\left\{
\begin{aligned}
&\Phi(\pmb\omega^{(t)},\mathcal{S}_k),&k&\in\mathcal{H} \\ 
&\pmb *,&k&\in\mathcal{B}, \label{Local-update-with-byzantine} \\ 
\end{aligned}
\right.
\end{equation}
where $\mathcal{H}=\mathcal{K} \backslash\mathcal{B}$ denotes the set of honest edge workers and $\Phi(\pmb\omega^{(t)},\mathcal{S}_k)$ denotes the normal local update process, which will be elaborated in Section \ref{SecII-C}.

\begin{figure}[t]
	\centering
	\includegraphics[width=0.45\textwidth]{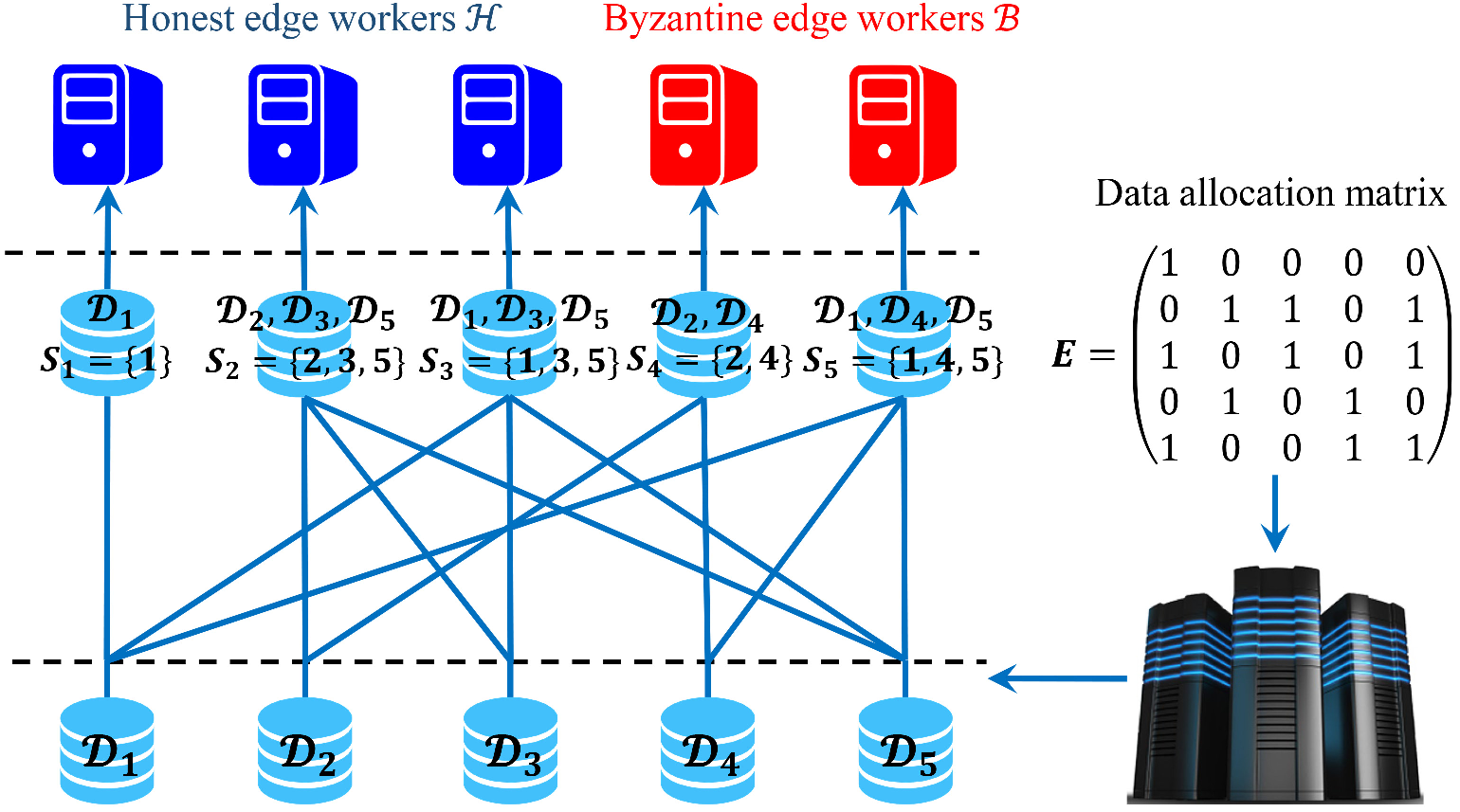}
	\caption{Data allocation process for edge workers in $\mathcal{K}$.}
	\label{fig:Data-allocation}
\end{figure}

\par
To quantify the Byzantine tolerance of the distributed data center network, we present the following definition \cite{Election-coding}:
\begin{definition}
Given the data allocation matrix $\pmb E$ and the corruption level $c$, the distributed data center network is $(c,\epsilon)$-Byzantine tolerant, if it can tolerate any types of attacks from $c$-fraction of Byzantine edge workers with at least probability $1-\epsilon$, i.e.,
\begin{equation}
\Pr(\hat{\mu}\neq\mu)\leq\epsilon,\;\;\text{with}\;\;0\leq\epsilon<\frac{1}{2}, \label{Byzantine-tolerant}
\end{equation}
where $\hat{\mu}$ denotes the output at the central edge server via Majority Vote and $\mu$ is the corresponding true output value. Notice that, the case of $\epsilon=0$ is called perfect Byzantine tolerance.
\end{definition}
From (\ref{Byzantine-tolerant}), we can see that a lower value of $\epsilon$ yields a higher security guarantee of the wireless data center network.

\subsection{Byzantine-Tolerant Scheme via Hierarchical Vote} \label{SecII-C}
To achieve $(c,\epsilon)$-Byzantine tolerance in the distributed wireless data center network, we develop a Byzantine-tolerant Hierarchical Vote scheme based on Bernoulli coding framework \cite{Election-coding}. Briefly, each edge worker is assigned to several sub-datasets randomly and derives its quantified one-bit gradient based on the allocated datasets, then the central edge server estimates the global model update by taking the majority opinion of the quantified gradients sent by all edge workers, even though some are malicious.

\par
As shown in Fig. \ref{fig:Data-allocation}, the key operation of Hierarchical Vote is the \emph{Data Allocation Process}, where each dataset $\mathcal{D}_i,\,\forall i\in\mathcal{K}$ is assigned to multiple edge workers based on the data allocation matrix $\pmb E$, which remains invariant during the whole learning process. Compared to the canonical distributed settings with $\pmb E=\pmb I_K$, the design of $\pmb E$ in Hierarchical Vote framework follows Bernoulli random coding scheme, where each entry $\pmb E_{ki}$ of $\pmb E$ is drawn independently according to $\text{Bernoulli}(p)$, where $p$ denotes the probability of $\pmb E_{ki}=1$. Essentially, the data allocation process can be considered as a dataset selection scheme where each edge worker independently selects each dataset in $\{\mathcal{D}_i\}_{i=1}^{K}$ with probability $p$. We note that $p$ is an adaptive parameter, which can be flexibly adjusted to tolerate different corruption levels $c$. Besides, in order to ensure that all edge workers participate in the learning process, we assume that all diagonal elements of $\pmb E$ are $1$, and the above data allocation scheme can be considered as a Byzantine-tolerant extension of the traditional distributed learning settings.

\begin{figure}[t]
	\centering
	\includegraphics[width=0.45\textwidth]{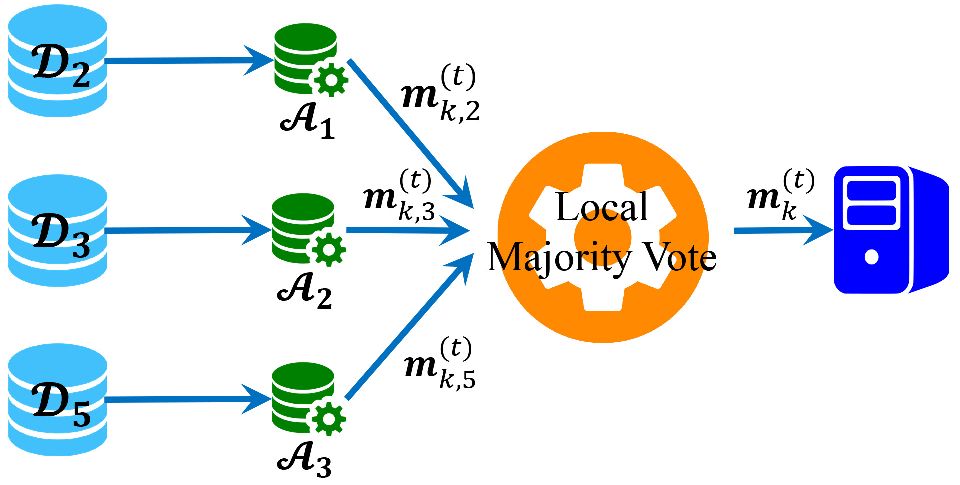}
	\caption{Local Majority Vote process of the $k$-th edge worker.}
	\label{fig:Local-update}
\end{figure}

\par
Based on the allocated datasets, each honest edge worker $k\in\mathcal{H}$ generates its local update message $\pmb m_k^{(t)}$ in parallel via \emph{Local Majority Vote}, i.e., $\Phi(\pmb\omega^{(t)},\mathcal{S}_k)$ defined in (\ref{Local-update-with-byzantine}). Specifically, as shown in Fig. \ref{fig:Local-update}, based on the allocated datasets $\{\mathcal{D}_i,\,i\in\mathcal{S}_k\}$, the $k$-th honest edge worker generates a random mini-batch $\mathcal{A}_{k,i}^{(t)}$ with $|\mathcal{A}_{k,i}^{(t)}|=A$ from each allocated datasets in $\mathcal{S}_k$, followed by computing $|\mathcal{S}_k|$ local stochastic gradients by using (\ref{Stochastic-gradient}), which are given by
\begin{equation}
\pmb g_{k,i}^{(t)}=\frac{1}{A}\sum\nolimits_{\pmb\xi\in\mathcal{A}_{k,i}^{(t)}}\nabla f(\pmb\omega^{(t)},\pmb\xi),\;\; i\in\mathcal{S}_k. \label{Stochastic-gradient-with-encode}
\end{equation}
Hence, the corresponding quantified one-bit local gradient set $\{\pmb m_{k,i}^{(t)}:i\in\mathcal{S}_k\}$ can be obtained via (\ref{Quantified-gradient}). Then, Majority Vote is performed locally to compress all the quantified gradients and obtain an encoded update message $\pmb m_k^{(t)}$, which is given by
\begin{equation}
\pmb m_k^{(t)}=\text{sign}\left(\sum\nolimits_{i\in\mathcal{S}_k}\pmb m_{k,i}^{(t)}\right). \label{Local-majority-voting}
\end{equation} 
On the other hand, the Byzantine edge workers in $\mathcal{B}$ can generate malicious update messages via (\ref{Local-update-with-byzantine}). Consequently, the central edge server aggregates all the local update messages $\{\pmb m_k^{(t)}\}_{k=1}^{K}$, followed by performing \emph{Global Majority Vote} to estimate the global model update via (\ref{Majority-vote}). Finally, the central edge server updates the global model $\pmb\omega^{(t)}$.

\par
In brief, the Hierarchical Vote framework achieves Byzantine tolerance by constructing a random data allocation scheme and two Majority Vote processes (one is manipulated locally to obtain the quantified local gradients at each edge worker, and the other is performed centrally to decode the global model update at the central edge server). Besides, we must note that the Hierarchical Vote framework is a redundant gradient computation-based scheme to achieve Byzantine tolerance, which is demonstrated to be an effective approach and widely used in the current Byzantine-tolerance related researches \cite{DETOX,DRACO,Election-coding}. Specifically, the redundant computation resources in the Hierarchical Vote framework include the cost for redundant data allocation at the central edge server, and the redundant computation for quantified local gradients $\{\pmb m_{k,i}^{(t)}:i\in\mathcal{S}_k\}$ at each edge worker in $\mathcal{H}$. In Section \ref{SecIII} and \ref{SecIV}, we will demonstrate that such redundant computation can exactly ensure Byzantine robustness via theoretical analysis and numerical simulations, respectively.

\subsection{Model Aggregation via Over-the-Air Computation}
According to the additive structure of (\ref{Majority-vote}), the Majority Vote decoding scheme completely fits the category of AirComp \cite{One-bit-AirComp1,FLAirComp1,FLAirComp2,FLAirComp4}, which can be leveraged to support a communication-efficient Byzantine-tolerant distributed learning system. To simplify the theoretical analysis, the downlink channels are assumed to be noise-free due to the less stringent power constraint at the base station \cite{FLPrivacy-free}, thereby each edge worker can receive the current global model $\pmb\omega^{(t)}$ without any distortion. We focus on the block flat-fading uplink channels whose channel coefficients remain invariant within one communication block. Besides, we assume that the channel coherence length allows the transmission of a quantified gradient vector, which is practically well justified when the $d$-dimensional model update $\pmb m_k^{(t)}$ is lightweight under a few tens of thousands of entries \cite{FLPrivacy-free}. Further, for the training of a high-dimensional model, the quantified gradients can be segmented and transmitted across multiple consecutive communication blocks \cite{Block-crossing}, and we leave this as our future work.

\par
In the $t$-th learning round, let $h_k^{(t)}\in\mathbb{C}$ denote the channel coefficient from the $k$-th edge worker to the central edge server. Owing to the recent advancements in channel estimation, we assume that perfect channel state information (CSI) is available among all edge workers for analytical ease \cite{Channel-inverse1,RIS}. Hence, we design the uniform-forcing transmit signal $\pmb x_{k}^{(t)}$ as follows:
\begin{equation}
\pmb x_k^{(t)}=\iota_{k}^{(t)}\pmb m_k^{(t)}=\rho^{(t)}\frac{\;(h_k^{(t)})^H}{|h_k^{(t)}|^2}\pmb m_k^{(t)}, \label{Transmit-signal}
\end{equation}
where $\iota_{k}^{(t)}=\rho^{(t)}/h_k^{(t)}$ denotes the transmit scalar, and $\rho^{(t)}$ represents the uniform power scaling factor, which is designed by the central edge server and broadcasted along with $\pmb\omega^{(t)}$ to all edge workers before each learning round. Given a maximum transmit power $P_0>0$, we can obtain the following power constraints for all edge workers in $\mathcal{K}$:
\begin{equation}
\big|\big|\pmb x_k^{(t)}\big|\big|^2\leq P_0,\;\;\forall k,t. \label{Power-constraint1}
\end{equation}
Based on the perfect CSI at each edge worker and recalling that $\pmb m_k^{(t)}\in\{-1,1\}^d$, the power constraints can be further simplified as follows:
\begin{equation}
\big(\rho^{(t)}\big)^2\leq\frac{P_0}{d}\big|h_k^{(t)}\big|^2,\;\;\forall k,t. \label{Power-constraint2}
\end{equation}

\par
Consequently, in the $t$-th learning round, all edge workers in $\mathcal{K}$ synchronously transmit the local updates $\{\pmb m_k^{(t)}\}_{k=1}^K$ to the central edge server, and the received aggregated signal at the central edge server is given by
\begin{equation}
\pmb y^{(t)}=\sum_{k=1}^{K}h_{k}^{(t)}\pmb x_{k}^{(t)}+\pmb n^{(t)}=\rho^{(t)}\sum_{k=1}^K\pmb m_{k}^{(t)}+\pmb n^{(t)}, \label{Aggregated-signal}
\end{equation}
where $\pmb n^{(t)}\in\mathbb{C}^{d}$ denotes the receiver noise, which can be modeled as i.i.d. additive white Gaussian noise (AWGN) according to $\mathcal{CN}(\pmb 0,N_0\pmb I_d)$. Then, the central edge server computes the processed update $\hat{\pmb r}^{(t)}$ via
\begin{equation}
\begin{split}
\hat{\pmb r}^{(t)}&=\text{Re}\left\{\pmb y^{(t)}\right\} \\
&=\underbrace{\rho^{(t)}\sum\nolimits_{k\in\mathcal{H}}\pmb m_k^{(t)}}_{\text{Normal updates}}+\underbrace{\rho^{(t)}\sum\nolimits_{k\in\mathcal{B}}\pmb m_k^{(t)}}_{\text{Corrupted updates}}+\underbrace{\text{Re}\left\{\pmb n^{(t)}\right\}}_{\text{AWGN}}. \label{Aggregated-signal-after-process}
\end{split}
\end{equation} 
Notice that, according to (\ref{Aggregated-signal-after-process}), except the corrupted model updates from Byzantine edge workers, the receiver noise $\pmb n^{(t)}$ can also cause inaccuracy in the model aggregation process. To quantify the influence of the receiver noise, we define SNR as the ratio of the maximum transmit power and the receiver noise power within one communication block, i.e., $\text{SNR}=P_0/d N_0$.

\par
An estimation of the global model update $\hat{\pmb v}^{(t)}$ can be obtained by performing Majority Vote defined in (\ref{Majority-vote}) at the central edge server, which is given by
\begin{equation}
\hat{\pmb v}^{(t)}=\text{sign}\left(\hat{\pmb r}^{(t)}\right)\in\{-1,1\}^d. \label{Central-majority-voting}
\end{equation}
Based on (\ref{Central-majority-voting}), the central edge server can update the global model and complete one learning round. We summarize the proposed AirComp-enabled Hierarchical Vote framework in the following \emph{Algorithm \ref{Algorithm-1}}.

\begin{algorithm}[t]
	\caption{AirComp-Enabled Hierarchical Vote Framework.}
	\KwIn{Global dataset $\mathcal{D}$, initial model parameter $\pmb\omega^{(0)}$, and the maximum number of learning rounds $T$.}
	\BlankLine
	\emph{Step 1: Data Allocation Process} \\
	The central edge server partitions $\mathcal{D}$ into $K$ subsets $\{\mathcal{D}_i\}_{i=1}^K$, and assigns them to all edge workers based on the designed Bernoulli coding matrix $\pmb E$. \\
	\BlankLine
	\emph{Step 2: Hierarchical Vote Process} \\
	\For{$t\leftarrow 0,1,\ldots,T-1$}{
		The central edge server broadcasts the current $\pmb\omega^{(t)}$. \\
		\For{each edge worker $k\in\mathcal{K}$ \emph{in parallel}}{
			\uIf{$k\in\mathcal{H}$}{
				Perform \emph{Local Majority Vote} via (\ref{Stochastic-gradient-with-encode}) and (\ref{Local-majority-voting}) to generate local updates. \\
			}
			\Else{
				Perform (\ref{Local-update-with-byzantine}) to generate malicious	updates.		
			}
		}
		The central edge server aggregates $\{\pmb m_k^{(t)}\}_{k=1}^{K}$ over-the-air and performs \emph{Global Majority Vote} via (\ref{Aggregated-signal})-(\ref{Central-majority-voting}). \\
		The central edge server updates the global model via $\pmb\omega^{(t+1)}=\pmb\omega^{(t)}-\eta\hat{\pmb v}^{(t)}$.
	}
	\BlankLine
	\KwOut{Global model parameter $\pmb\omega^{(T)}$.
	}\label{Algorithm-1}
\end{algorithm}

\par
\emph{Remark 1 (Resilience to other types of attacks)}: We discuss a more comprehensive Byzantine attack model concerning attacks with different frequencies and attacks form illegitimate edge workers. The former can be defended by systematic synchronization, where the central edge server can distinguish and exclude the Byzantine edge workers that upload corrupted messages with different frequencies. Besides, there also exist multiple approaches to defend asynchronous attacks, such as blockchain \cite{Blockchain}, and we leave this extension as our future work. The latter can be tackled by utilizing the celebrated direct sequence spread spectrum (DSSS) technique \cite{FLAirComp2}. Specifically, all the legitimate edge workers share a common spreading code designed by the central edge server and the illegitimate workers do not. During the model aggregation process, the central edge server can decode the aggregated signal based on the spreading code while the negative influence from illegitimate edge workers can be suppressed due to the absence of correct spreading code.

\par
\emph{Remark 2 (Synchronization in wireless data center networks)}: The deployment of the proposed framework entails synchronization among edge workers, which can be achieved in the current data center networks. Particularly, from the perspective of wireless communication, a common synchronous approach is the timing advance scheme \cite{Timing-advance} in 4G long-term evolution (LTE) and 5G new radio (NR), which only requires a tolerable timing synchronization offset with $0.1\mu s$ shorter than the typical length of cyclic prefix with $5\mu s$ and this offset can be compensated by channel equalization. Besides, the current data center networks also provide technique insurance for synchronization \cite{Syn-datacenter}. Specifically, the advanced hardwares such as Remote Direct Memory Access (RDMA) and Solid State Drives (SSDs), the advancements in synchronous algorithms, and the progress in cyber-physical systems pave the road to achieving the full synchronization of partial data center networks, that is, all the edge workers and the edge server.

\section{Performance Analysis} \label{SecIII}
In this section, we first theoretically analyze the inherent relationship among system security, corruption level, and the receiver noise, followed by characterizing the convergence behavior of the proposed AirComp-enabled Hierarchical Vote framework under the non-convex settings. Further, we conduct theoretical analysis and numerical simulations to demonstrate the advancements of the proposed framework in communication and computation efficiency in the wireless data center networks.

\subsection{Assumptions and Preliminaries}
We first list several widely used assumptions \cite{SignSGD1,SignSGD2} for the analysis of the one-bit Byzantine-tolerant system. Assumptions \ref{Assumption-1}, \ref{Assumption-2}, \ref{Assumption-3} are commonly leveraged in proving the convergence of non-convex optimization problems, and Assumption \ref{Assumption-4} characterizes the distribution properties of the data-stochasticity induced stochastic gradient noise.
\begin{assumption} \label{Assumption-1}
	(Smoothness): The global loss function $F(\pmb\omega)$ defined in (\ref{Global-loss-function}) is Lipschitz continuous with a vector of non-negative constants $\pmb L=[L_1,\ldots,L_d]$, which indicates that for any $\pmb\omega$, $\pmb\omega'$, the following inequality holds
	\begin{equation}
	\left|F(\pmb\omega')-F(\pmb\omega)-\pmb g^T(\pmb\omega'-\pmb\omega)\right|\leq\frac{1}{2}\sum_{j=1}^{d}L_j\left(\pmb\omega'(j)-\pmb\omega(j)\right)^2, \nonumber
	\end{equation}
	where $\pmb\omega(j)$ and $\pmb\omega'(j)$ denote the $j$-th entry of $\pmb\omega$ and $\pmb\omega'$, respectively, and $\pmb g=\nabla F(\pmb\omega)$.
\end{assumption}

\begin{assumption} \label{Assumption-2}
	(Global Loss Lower Bound): For any model parameter $\pmb\omega\in\mathbb{R}^d$, the global loss function $F(\pmb\omega)$ defined in (\ref{Global-loss-function}) has a lower bound $F^*$, i.e., $F(\pmb\omega)\geq F^*,\,\forall\pmb\omega$, which ensures convergence of the non-convex optimization problem to a stationary point.
\end{assumption}

\begin{assumption} \label{Assumption-3}
	(Unbiasedness and Variance Bound): The stochastic gradient $\tilde{\pmb g}$ generated from any training data sample $\pmb\xi\in\mathcal{D}$ is an unbiased estimate on the true gradient $\pmb g=\nabla F(\pmb\omega)$ with element-wise variance bound, i.e.,
	\begin{subequations}
	\begin{align}
	&\mathbb{E}\big[\tilde{\pmb g}(j)-\pmb g(j)\big]=0,\quad\forall j\in[d], \\
	&\mathbb{E}\big[\tilde{\pmb g}(j)-\pmb g(j)\big]^2\leq\sigma_j^2,\quad\forall j\in[d], \label{SGD-variance-bound}
	\end{align}
	\end{subequations}
	where $\tilde{\pmb g}(j)$ and $\pmb g(j)$ denote the $j$-th entry of $\tilde{\pmb g}$ and $\pmb g$, respectively, and $\pmb\sigma=[\sigma_1,\ldots,\sigma_d]^T$ is a vector containing the variance bound for each entry. 
\end{assumption}

\begin{assumption} \label{Assumption-4}
	(Stochastic Gradient Noise): Given the model parameter $\pmb\omega$ and any training data sample $\pmb\xi\in\mathcal{D}$, the element-wise stochastic gradient noise denoted by $\tilde{\pmb g}(j)-\pmb g(j)$ has a unimodal distribution, which is symmetric about the mean \cite{Election-coding,SignSGD1,One-bit-AirComp1}.
\end{assumption}
By integrating Assumptions \ref{Assumption-3} and \ref{Assumption-4}, we note that the mini-batch gradient $\pmb g_{k,i}^{(t)}(j)$ defined in (\ref{Stochastic-gradient-with-encode}) follows a unimodal distribution which is symmetric about the true gradient $\pmb g^{(t)}(j)$. The element-wise mean and variance satisfy the following relationships \cite{SignSGD1,SignSGD2}:
\begin{subequations}
	\begin{align}
	&\mathbb{E}\big[\pmb g_{k,i}^{(t)}(j)-\pmb g^{(t)}(j)\big]=0,\quad\forall j,i,k,t, \\
	&\mathbb{E}\big[\pmb g_{k,i}^{(t)}(j)-\pmb g^{(t)}(j)\big]^2\leq\sigma_j^2/A, \quad\forall j,i,k,t, \label{Mini-batch-SGD-variance-bound}
	\end{align}
\end{subequations}
where $A$ denotes the size of the mini-batch. Further, for analytical ease, we present the following definition of gradient-signal-to-data-noise ratio (GSNR) \cite{SignSGD1,SignSGD2,One-bit-AirComp1}:
\begin{definition} \label{Definition-2}
For a single entry $j\in[d]$, the GSNR $J_j$ of the mini-batch stochastic gradient $\pmb g_{k,i}^{(t)}(j)$ is defined as:
\begin{equation}
J_j=\sqrt{A}\,\frac{\big|\pmb g_{k,i}^{(t)}(j)\big|}{\sigma_j}. \label{Stochastic-SNR}
\end{equation}
\end{definition}
\noindent
From (\ref{Stochastic-SNR}), it is simple to verify that a larger size of mini-batch $A$ leads to a higher value of GSNR $J_j$.

\subsection{AirComp-Enabled Byzantine-Tolerant Error Bound} \label{SecIII-B}
As elaborated in Section \ref{SecII-C}, the Byzantine-tolerant Hierarchical Vote framework consists of two Majority Vote processes where one is performed locally at each honest edge worker to determine the signs of local gradients and the other is operated over-the-air to obtain the global majority opinions at the central edge server. For analytical ease, we focus on the element-wise error probability $P_{\text{err}}^{(t)}$ (or equivalently, decoding bit error probability), which quantifies the error probability when the central edge server estimates the sign of $\pmb g^{(t)}(j),\,\forall j\in[d]$.

\par
We first analyze the estimation error probability of \emph{Local Majority Vote} process, where the error comes from the stochastic gradient noise of the honest edge workers. Specifically, according to (\ref{Local-majority-voting}), our objective is to measure the element-wise estimation error probability $\pmb q^{(t)}(j)\triangleq\Pr\big[\pmb m_k^{(t)}(j)\neq\text{sign}\big(\pmb g^{(t)}(j)\big)\big]$ when locally determining the signs of $\pmb g^{(t)}(j)$. To this end, for any honest edge worker $k\in\mathcal{H}$ and given the indexes of the allocated datasets $\mathcal{S}_k$, we present the following proposition to provide an upper bound of the conditional local estimation error $\pmb q_{\mathcal{S}_k}^{(t)}(j)\triangleq\Pr\big[\pmb m_k^{(t)}(j)\neq\text{sign}\big(\pmb g^{(t)}(j)\big)\,\big|\,\mathcal{S}_k\big]$.
\begin{proposition} \label{Proposition-1}
Suppose $\mathcal{S}_k$ denote the index set of datasets allocated to an honest edge worker $k\in\mathcal{H}$, the upper bound probability of this edge worker transmitting a wrong sign $\pmb m_k^{(t)}(j)$ to the central edge server is given as:
\begin{equation}
\pmb q_{\mathcal{S}_k}^{(t)}(j)=\Pr\left(\pmb m_k^{(t)}(j)\neq\text{sign}\left(\pmb g^{(t)}(j)\right)\,\big|\,\mathcal{S}_k\right)\leq\frac{1}{\,J_j\sqrt{|\mathcal{S}_k|}\,}, \label{Proposition-1-relationship}
\end{equation}
where $J_j$ denotes the element-wise GSNR defined in (\ref{Stochastic-SNR}).
\end{proposition}
\begin{proof}
Please refer to Appendix \ref{App-1}.
\end{proof}
As revealed in (\ref{Proposition-1-relationship}), the error probability bound of the case $|\mathcal{S}_k|=1$ (or equivalently, the canonical distributed settings with $p=0$) is $1/J_j$, which indicates that the canonical distributed settings cannot guarantee a right sign with high probability, yielding its inherent vulnerability to errors and attacks. In comparison, as the number of allocated datasets $|\mathcal{S}_k|$ increases, the edge worker $k\in\mathcal{H}$ can obtain a correct estimation on the sign of $\pmb g^{(t)}(j)$ with a higher probability. Intuitively, the high estimation probability is obtained at the expense of the redundant computation resources consumed on each dataset in $\mathcal{S}_k$. Moreover, according to (\ref{Stochastic-SNR}), a larger mini-batch size $A$ results in a higher value of GSNR $J_j$, thereby also contributing to the accuracy of estimation.

\par
Further, we analyze the impact of \emph{Data Allocation Process} on the element-wise estimation error probability $\pmb q^{(t)}(j)$. As elaborated in Section \ref{SecII-C}, $\mathcal{S}_k$ consists of a fixed allocated dataset $\mathcal{D}_k$ and $n_k$ redundantly assigned datasets, i.e., $|\mathcal{S}_k|=1+n_k$. We note that each redundant dataset is allocated with a well-designed probability $p$, yielding a random variable $n_k$ according to $\text{Binomial}(K-1,p)$. Based on this, we present the following theorem to explore the element-wise local estimation error probability $\pmb q^{(t)}(j)$ for an honest edge worker $k\in\mathcal{H}$.
\begin{theorem} \label{Theorem-1}
Given a coordinate $j\in[d]$, suppose $p\in(\frac{4}{J_j^2K},1]$ denote the probability of allocating dataset $\mathcal{D}_i,\,\forall i\neq k$ to edge worker $k$, i.e., $\pmb E_{ki}=1$. The estimation error probability of an honest edge worker $k$ transmitting a wrong sign to the central edge server is upper-bounded by: 
\begin{equation}
\pmb q^{(t)}(j)=\Pr\left(\pmb m_k^{(t)}(j)\neq\text{sign}\left(\pmb g^{(t)}(j)\right)\right)\leq\frac{1}{\,J_j\sqrt{Kp}\,}. \label{Theorem-1-relationship}
\end{equation}
\end{theorem}
\begin{proof}
Please refer to Appendix \ref{App-2}.
\end{proof}
The error probability bound of $\pmb q^{(t)}(j)$ in \emph{Theorem \ref{Theorem-1}} is an average of $\pmb q_{\mathcal{S}_k}^{(t)}(j)$ over $n_k$. We note that as the allocation probability $p$ increases, the element-wise local estimation error probability bound will decrease. Essentially, similar to \emph{Proposition \ref{Proposition-1}}, \emph{Theorem \ref{Theorem-1}} indicates that a larger $p$ leads to more allocated datasets, yielding a higher probability of edge worker $k$ to correctly estimate the sign of $\pmb g^{(t)}(j)$.

\par
We turn to analyze the estimation error probability $P_{\text{err}}^{(t)}\triangleq\Pr\big[\pmb v^{(t)}(j)\neq\text{sign}(\pmb g^{(t)}(j))\big]$ of \emph{Global Majority Vote} process in the presence of Byzantine edge workers and the receiver noise, where the central edge server performs over-the-air Majority Vote to obtain an estimation of the global model update. To measure the estimation performance of the AirComp-enabled Hierarchical Vote system, we consider the scenario where the global estimation bit error probability $P_{\text{err}}^{(t)}$ is maximized. In particular, the Byzantine edge workers in $\mathcal{B}$ collude with each other and send the inverse signs of the true gradient, i.e., $\pmb m_k^{(t)}(j)=-\text{sign}\left[\pmb g^{(t)}(j)\right],\,\forall k\in\mathcal{B}$ \cite{Election-coding}. Besides, to ensure the attacks take effect, we assume that the receiver noise $\pmb n^{(t)}$ cannot reverse the signs of $\pmb m_k^{(t)}$ from the Byzantine edge workers. According to (\ref{Aggregated-signal}), the aggregated signal consists of the update messages from honest and Byzantine edge workers and a scaled version of the receiver noise. In the $t$-th learning round, the following theorem provides an upper bound of the global decoding bit error probability $P_{\text{err}}^{(t)}$ at the central edge server.
\begin{theorem} \label{Theorem-2}
Given the allocation probability $p$, the element-wise local estimation bit error probability $\pmb q^{(t)}(j)$ can be upper-bounded by (\ref{Theorem-1-relationship}). If the corruption level $c$ satisfies
\begin{equation}
(1-c)\left(1-\pmb q^{(t)}(j)\right)>\frac{1}{2},\;\;\forall j\in[d], \label{Necessary-condition}
\end{equation}
the global decoding bit error probability at the central edge server is upper-bounded by:
\begin{equation}
\begin{split}
P_{\text{err}}^{(t)}&=\Pr\left(\text{sign}\big(\tilde{\pmb r}^{(t)}(j)\big)\neq\text{sign}\big(\pmb g^{(t)}(j)\big)\right) \\
&\leq\frac{1}{2}\sqrt{\frac{1-c}{K}}+\frac{1}{K\rho^{(t)}}\sqrt{\frac{N_0}{2}}\leq \epsilon,\;\;\forall j\in[d], \label{Theorem-2-relationship}
\end{split}
\end{equation}
where $\epsilon\in[0,\frac{1}{2})$ denotes a given error probability bound, which indicates that the AirComp-enabled Hierarchical Vote framework can achieve $(c,\epsilon)$-Byzantine tolerance.
\end{theorem}
\begin{proof}
Please refer to Appendix \ref{App-3}.
\end{proof}
From \emph{Theorem \ref{Theorem-2}}, the decoding bit error probability $P_{\text{err}}^{(t)}$ is related to the Byzantine attacks and the wireless environment. Conditioned on (\ref{Necessary-condition}), we note that the bound in (\ref{Theorem-2-relationship}) is a decreasing function about the number of edge workers $K$. For a sufficiently large value of $K$, the proposed framework can achieve arbitrarily small $P_{\text{err}}^{(t)}$, and (\ref{Necessary-condition}) reduces to $1-c\rightarrow 1/2$ since $\pmb q^{(t)}(j)\rightarrow 0$, indicating that the proposed framework can achieve robustness when the corruption level $c<1/2$ in the asymptotic regime of large $K$. Besides, according to (\ref{Theorem-1-relationship}) and (\ref{Necessary-condition}), a larger value of $p$ yields a higher level of Byzantine tolerance, i.e., a larger value of $c$, at the expense of redundant computation resources. The influence of wireless channel is revealed in the second term of (\ref{Theorem-2-relationship}). According to (\ref{Power-constraint2}), $\rho^{(t)}$ is restricted by the edge worker with the weakest channel response, i.e., the straggler, within the $t$-th learning round. Specifically, a straggler with $|h_k^{(t)}|\approx 0$ results in a small value of $\rho^{(t)}$, which degrades the receiver SNR and yields a high error probability.

\subsection{Convergence Analysis on Non-Convex Settings}
To support the AirComp-enabled Byzantine-tolerant framework in the widely used deep neural networks (DNN), we focus on the analysis of the convergence behavior in non-convex settings, i.e., the empirical global loss function $F(\pmb\omega)$ defined in (\ref{Global-loss-function}) is non-convex. As revealed in \cite{SignSGD1}, the expectation of gradient norm is used as an indicator of convergence for non-convex settings. Particularly, the proposed framework achieves a $\tau$-suboptimal solution if $\mathbb{E}\left[\frac{1}{T}\sum_{t=1}^{T}\left|\left|\nabla F(\pmb\omega^{(t-1)})\right|\right|_1^2\right]\leq\tau$ holds, which guarantees the convergence of the proposed framework to a stationary point. The convergence result of the non-convex scenario is given in the following theorem.
\begin{theorem} \label{Theorem-3}
Under Assumptions \ref{Assumption-1}, \ref{Assumption-2}, \ref{Assumption-3}, and \ref{Assumption-4}, given the learning rate $\eta=1\big/\sqrt{T||\pmb L||_1}$. Suppose that the corruption level $c$ satisfies (\ref{Necessary-condition}), then the averaged expected gradient norm after $T$ learning rounds can be upper-bounded by
\begin{equation}
\mathbb{E}\left[\frac{1}{T}\sum_{t=0}^{T-1}\big|\big|\pmb g^{(t)}\big|\big|_1\right]\leq\frac{1}{\sqrt{T}}\frac{\sqrt{||\pmb L||_1}}{\Delta}\left(F(\pmb\omega^{(0)})-F^*+\frac{1}{2T}\right), \label{Theorem-3-relationship}
\end{equation}
where
\begin{equation}
\Delta=\left(1-\sqrt{\frac{1-c}{K}}-\frac{\sqrt{2}}{K\rho_{\text{min}}}\sqrt{N_0}\right), \nonumber
\end{equation}
with
\begin{equation}
\rho_{\text{min}}=\underset{t\in[T]}{\min}\,\rho^{(t)}\geq\frac{\sqrt{2N_0}}{K\big[1-\sqrt{(1-c)/K}\big]}. \label{Satisfactory-channel-condition}	
\end{equation}
\end{theorem}
\begin{proof}
Please refer to Appendix \ref{App-4}.
\end{proof}
From \emph{Theorem \ref{Theorem-3}}, conditioned on (\ref{Necessary-condition}) and (\ref{Satisfactory-channel-condition}), we observe that the averaged expected gradient norm converges to zero as the number of learning rounds $T$ increases, indicating that the proposed framework leads the non-convex distributed learning task to a stationary point with a sufficient number of learning rounds, despite the existence of the Byzantine edge workers and the receiver noise. Besides, the bound in (\ref{Theorem-3-relationship}) is also a decreasing function of $K$, indicating that a larger value of $K$ contributes to the convergence speed, which confirms the previous conclusions. Moreover, similar to (\ref{Theorem-2-relationship}), the influence of channel fading is reflected in $\rho^{(t)}$. In particular, the straggler issue increases the global decoding error probability, thereby degenerating the convergence rate. To alleviate this issue, multiple effective approaches have been proposed. Specifically, a binary device selection scheme specialized in \cite{FLAirComp2} excludes the edge workers based on a preset threshold at the expense of reducing the volume of data. Further, reconfigurable intelligent surface technique turns out to be a promising technology to support fast and reliable model aggregation by reconfiguring the channel propagation environment \cite{RIS}. We leave the extension to tackle the unfavorable channel propagations as our future work.

\begin{figure}[t]
	\centering
	\begin{minipage}[p]{0.45\textwidth}
		\centering
		\includegraphics[width=0.90\textwidth]{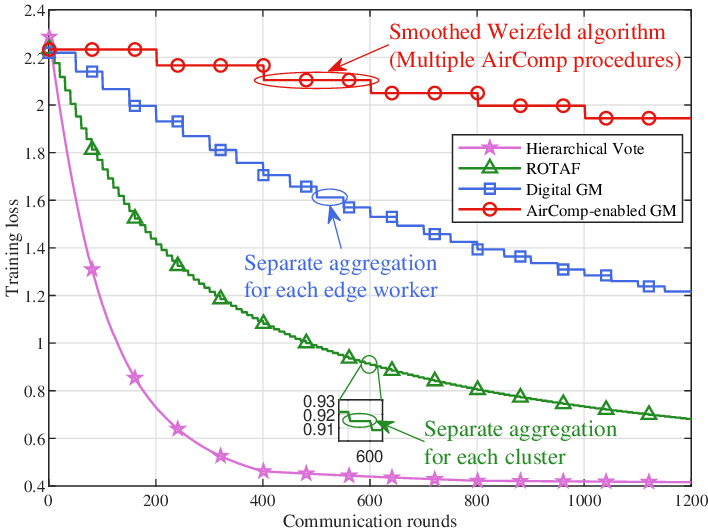}
	\end{minipage}
	\begin{minipage}[p]{0.45\textwidth}
		\centering
		\includegraphics[width=0.90\textwidth]{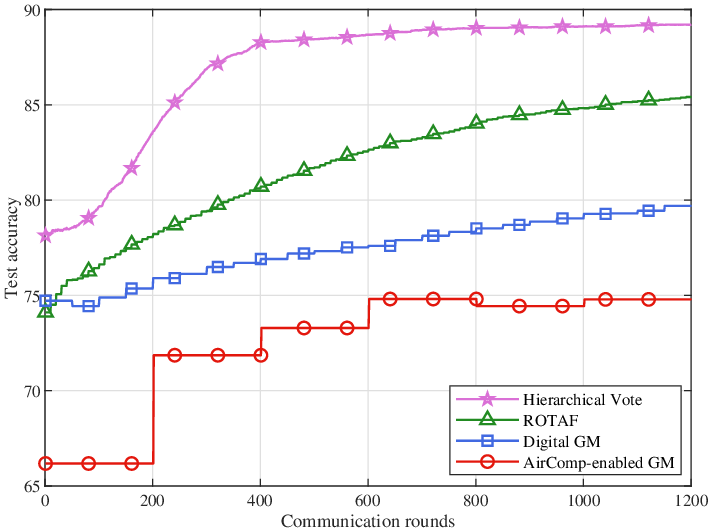}
	\end{minipage}
	\caption{Comparison of the four Byzantine-tolerant schemes in communication and computation efficiency. The default parameters are $K=50$, $B=5$, $p=0.1$, and $U=200$.}
	\label{fig:Comm-comp-eff}
\end{figure}

\subsection{Analysis of Communication and Computation Efficiency}
We comprehensively compare the Hierarchical Vote with the existing approaches \cite{FLRobust-GM2,Byzantine-AirComp1,Byzantine-AirComp2} concerning the communication and computation efficiency. From a theoretical perspective, AirComp outperforms the traditional digital communication schemes, which has been demonstrated in \cite{FLAirComp2}. Specifically, suppose $T_{Trad}$ and $T_{AirComp}$ denote the communication latency of the traditional digital schemes and AirComp, respectively. The communication latency ratio for each learning round is given by $r=\frac{T_{Trad}}{T_{AirComp}}\approx O(K)$, which indicates that $r$ scales approximately linearly with the number of edge workers. Particularly, if the number of participated edge workers $K$ increases, the traditional digital schemes will entail a significant communication latency compared with AirComp, which demonstrates the communication efficiency of AirComp. Besides, AirComp further reduces computation latency by integrating communication and computation without entailing additional resources for digital operations (for short, Digital) including encoding, decoding, and post-computing process.

\par	
Secondly, we compare the proposed framework with the existing Byzantine-tolerate schemes via numerical simulations. Given the total number of communication rounds, Fig. \ref{fig:Comm-comp-eff} shows that: 1) The convergence rate of AirComp-enabled GM is the lowest due to the multiple over-the-air aggregations for the convergence of smoothed Weiszfeld algorithm in each learning round, but it requires no high-cost post-processing procedures. 2) Digital GM separately aggregates the updates from all edge workers in a digital manner and its convergence rate declines with the increase of $K$. 3) ROTAF divides the edge workers evenly into $G$ clusters and separately receives the updates from all clusters via AirComp to perform GM. Thus, it can be seen that the inner aggregation process declines the convergence rate of the aforementioned schemes while the proposed framework avoids this and enjoys high communication efficiency. Table \ref{Table:comp-oper} lists the number of the computation operations that required for each scheme in one learning round. For instance, $1\times 50$ in the table indicates that in AirComp-enabled GM, all the $K=50$ edge workers are required to complete $1$ step of Local SGD in parallel during one learning round. The computational complexities of Local SGD and GM are $O(Ad)$ and $O(UKd)$ \cite{FLRobust-GM1}, respectively, where $U$ refers to the maximum iterations for the convergence of smoothed Weiszfeld algorithm satisfying $UK\gg A$. The computation cost becomes intolerable especially when the model dimension $d$ is large. Hence, due to the fact that the required cost for mini-batch gradient computation is far lower than that of GM and the edge workers in data center networks have sufficient computing resources to complete the multiple gradient computation operations in parallel, it is reasonable that the proposed framework outperforms the GM-based schemes in computation efficiency. Consequently, the proposed framework can be considered as a communication- and computation-efficient scheme in wireless data center networks.

\begin{table}[h]\scriptsize
	\centering
	\begin{tabular}{|>{}c ||>{}c|>{}c|>{}c|>{}c|}
		\hline
		\diagbox{\textbf{Schemes}}{\textbf{Operations}} & \textbf{Local SGD} & \textbf{GM} & \textbf{AirComp} &
		\textbf{Digital} \\
		\hline
		\textbf{AirComp-enabled GM} & $1\times 50$ & $0$ & $200$ & $0$ \\
		\hline
		\textbf{Digital GM} & $1\times 50$& $1$ & $0$ & $50$ \\
		\hline
		\textbf{ROTAF} & $1\times 50$ & $1$ & $10$ & $0$ \\
		\hline
		\textbf{Hierarchical Vote} & $5\times 50$ & $0$ & $1$ & $0$ \\
		\hline
	\end{tabular}
    \caption{The number of computation operations required for one learning round of the compared schemes.}
    \label{Table:comp-oper}
\end{table}

\section{Numerical Simulation} \label{SecIV}
In this section, we conduct numerical simulations to gain insights into the advantages of the AirComp-enabled Hierarchical Vote framework in the presence of Byzantine attacks.

\subsection{Simulation Settings}
We consider a wireless data center network consisting of one central edge server and $K=50$ edge workers. The wireless channel coefficients $\{h_k^{(t)}\}$ between the central edge server and each edge worker are assumed to be distributed according to $\mathcal{CN}(0,1)$. The transmit SNR is set to be $10\,\text{dB}$. To account for Byzantine attacks, the corruption level $c$ is set to be $0.4$, i.e., there exist $30$ honest and $20$ Byzantine edge workers in this system. We consider the image classification task based on the well-known MNIST dataset, which contains $10$ classes of handwritten digits ranging from $0$ to $9$. We test the learning performance of the proposed algorithm in the presence of the following four types of Byzantine attacks, where the first type of attack represents the data poisoning attacks and the others indicate the model poisoning attacks:
\begin{itemize}
\item \emph{Label-flipping attack}: The Byzantine edge workers can reverse the labels of their local datasets \cite{Byzantine-AirComp1,Coordinate-wise-trimmed-mean}. Specifically, the Byzantine edge workers replace the label $y$ of each training data point with $9-y$. 
\item \emph{Mimic attack}: All Byzantine edge workers copy the outputs from a certain edge worker, which enhances data heterogeneity by over-emphasizing the updates from one worker and burying the others \cite{Mimic}.
\item \emph{Directional attack}: The Byzantine edge workers transmit an all-one vector to guide the global model to a certain direction \cite{Election-coding,DRACO}. 
\item \emph{Omniscient attack}: The Byzantine edge workers collude with each other and adaptively modify their model updates as $\pmb m_k^{(t)}=-\text{sign}\left(\sum_{k\in\mathcal{H}}\pmb m_k^{(t)}\right)$, so that the received signals at the central edge server tend to zero \cite{FLRobust-GM1,FLByzantine-SGD3}.
\end{itemize}
We measure the learning performance in terms of (cross-entropy) training loss evaluated over the training samples and test accuracy derived from the test dataset, with respect to the number of communication rounds. To illustrate the advancement of the proposed Hierarchical Vote framework, we compare the learning performance of the proposed \emph{Algorithm \ref{Algorithm-1}} with Majority Vote-enabled SignSGD in the presence and absence of Byzantine attacks, respectively. Besides, we compare the learning performance of \emph{Algorithm \ref{Algorithm-1}} with its noise-free version to illustrate the impact of receiver noise.

\begin{figure*}
	\centering
	\begin{minipage}{\linewidth}
		\includegraphics[scale=0.43]{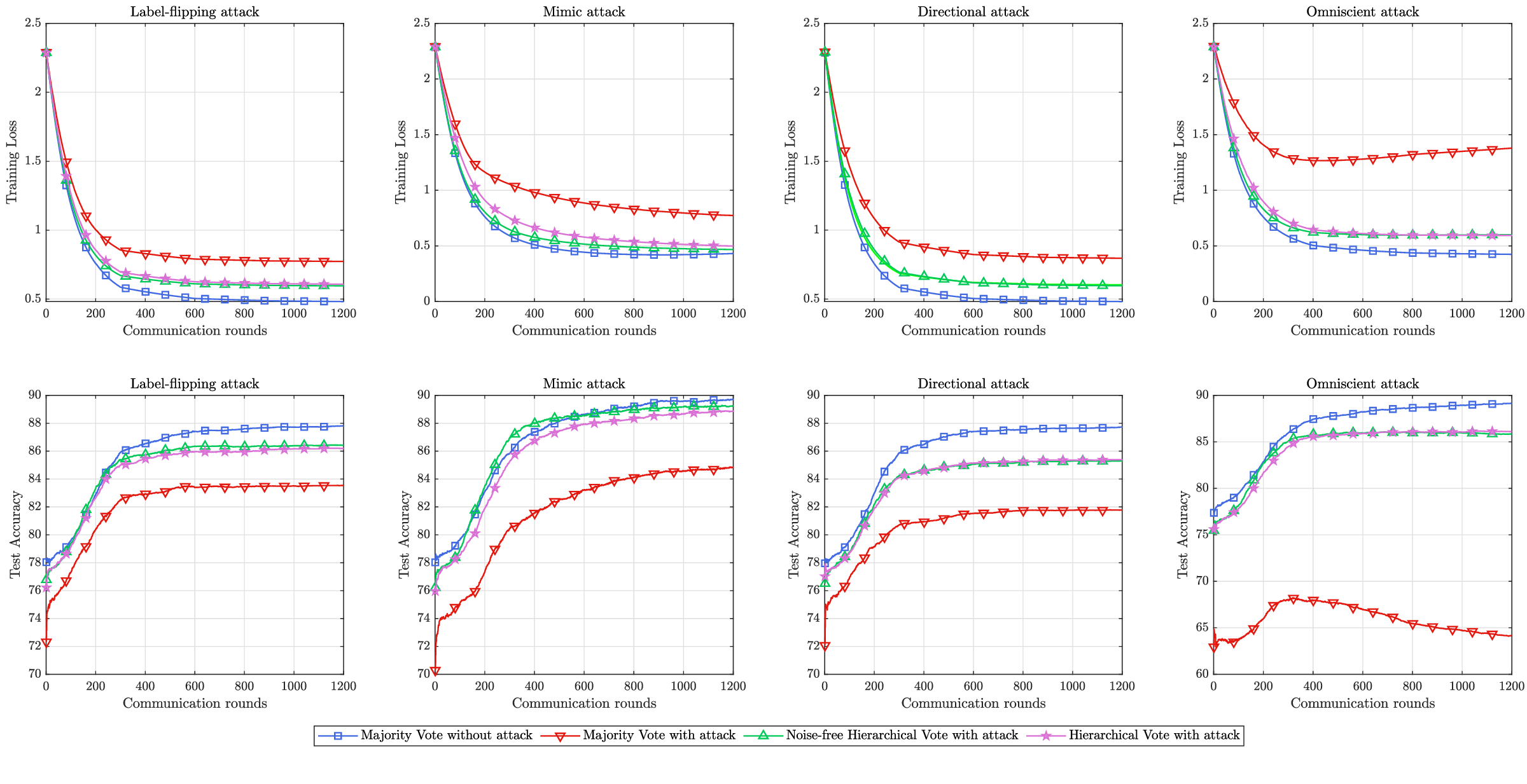}
		\vskip -10pt
		\caption{The learning performance of the proposed scheme in the presence of Byzantine attacks under convex settings.}
		\label{fig:Logistic-results}
	\end{minipage}
\end{figure*}

\subsection{Convex Settings: Logistic Regression}
We train logistic regression models on the MNIST dataset for convex settings under the mentioned four types of attacks, respectively. The corresponding training loss and test accuracy are plotted in Fig. \ref{fig:Logistic-results}. Specifically, we note that the learning performance of the Majority Vote-enabled SignSGD greatly degenerates in the presence of Byzantine attacks. By comparison, the proposed AirComp-enabled Hierarchical Vote scheme achieves approximately the same learning performance as the SignSGD schemes without attacks, which demonstrates the effectiveness of the proposed Hierarchical Vote scheme in tolerating multiple types of Byzantine attacks. Besides, by comparing the learning performance of the Hierarchical Vote with/without noise, we can conclude that the Majority Vote process at the central edge server is capable of achieving robustness to the inevitable receiver noise, which matches the conclusions in \cite{One-bit-AirComp1}.

\par
Fig. \ref{fig:Probability-results} shows the impact of allocation probability $p$ on the tolerance of the proposed AirComp-enabled Hierarchical Vote scheme against various Byzantine attacks. From a Byzantine tolerance perspective, we note that the learning performance of the proposed Hierarchical Vote scheme enhances as the allocation probability $p$ increases, which indicates that a higher value of $p$ yields a higher Byzantine tolerance of the proposed scheme and demonstrates the theoretical results in \emph{Theorem \ref{Theorem-1}}. From a convergence rate perspective, we note that a larger value of $p$ can improve the convergence rate, which comes from the redundant computational resource. Specifically, as revealed in \emph{Theorem \ref{Theorem-1}}, with the increase of $p$, each honest edge worker can estimate the true global update with a high probability, yielding a fast convergence rate.

\begin{figure*}[t]
	\begin{minipage}{\linewidth}
		\includegraphics[scale=0.43]{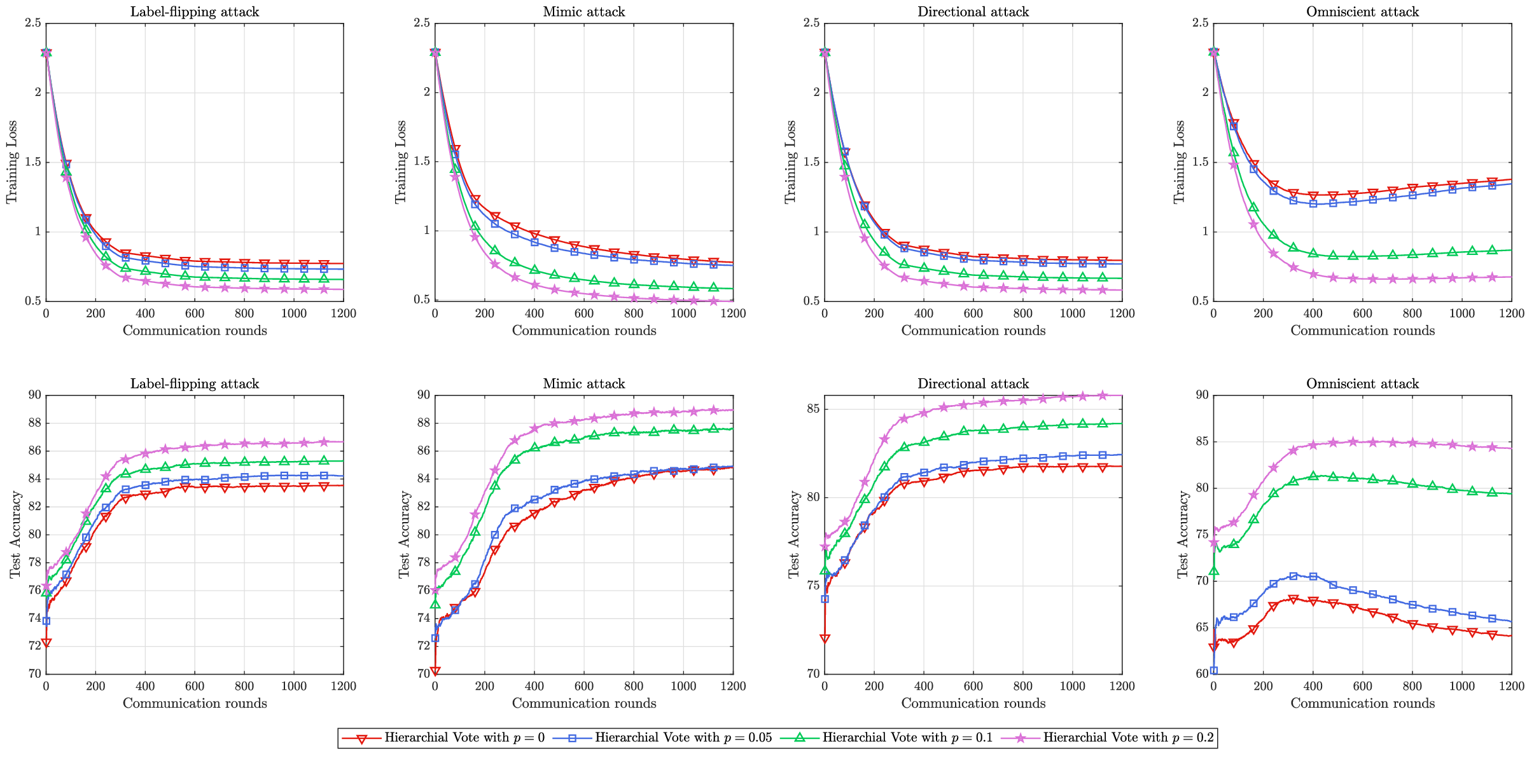}
		\vskip -10pt
		\caption{The impact of allocation probability $p$ on the learning performance under different types of Byzantine attacks.
		}
		\label{fig:Probability-results}
	\end{minipage}
	\vskip 15pt
	\begin{minipage}{\linewidth}
		\includegraphics[scale=0.43]{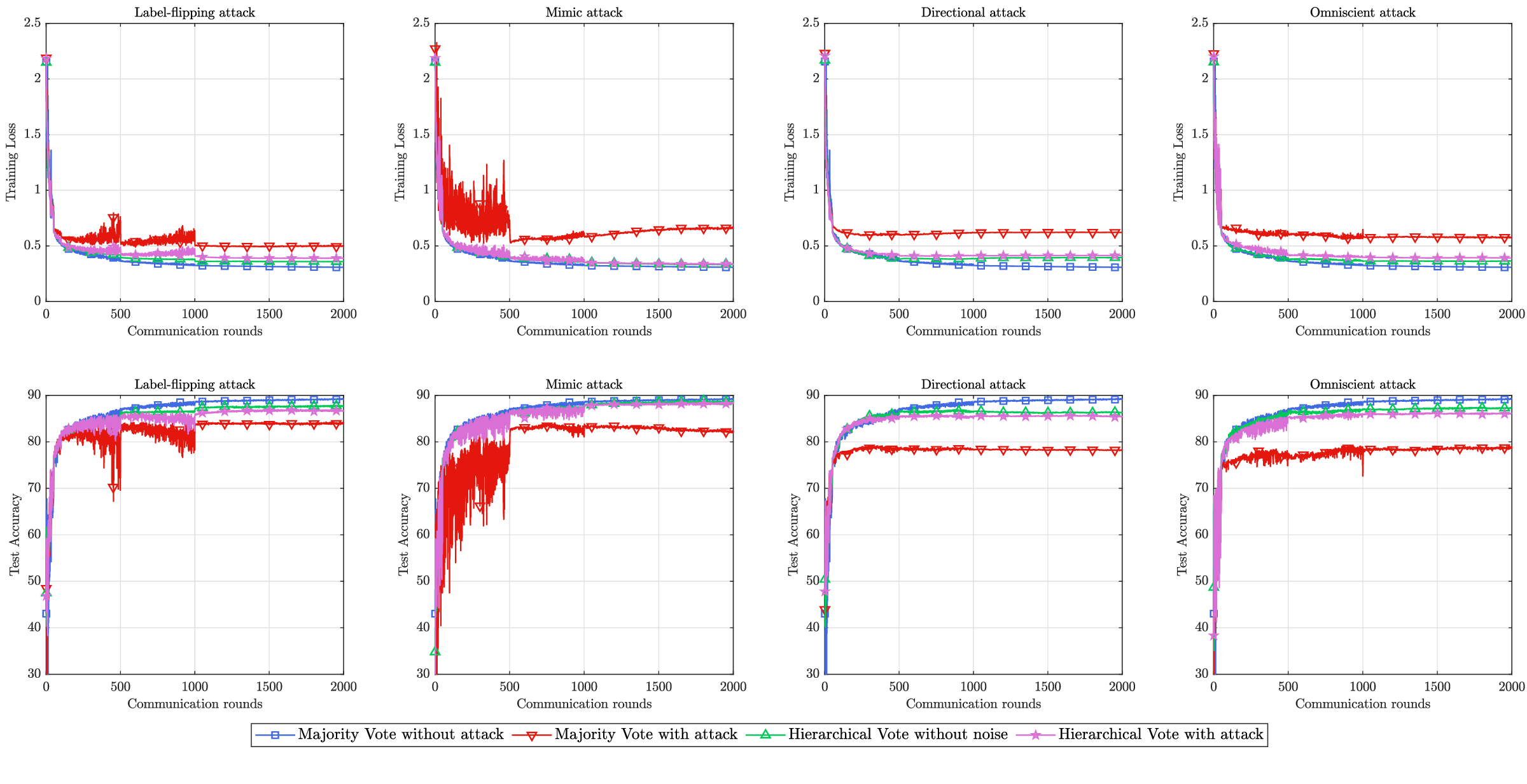}
		\vskip -10pt
		\caption{The learning performance of the proposed scheme in the presence of Byzantine attacks under non-convex settings.}
		\label{fig:CNN-results}
	\end{minipage}
\end{figure*}

\subsection{Non-Convex Settings: Convolution Neural Networks}
For non-convex settings, we train convolution neural network (CNN) models on the Fashion MNIST dataset under the four types of attacks, respectively. We modify the transmit SNR to be $5\,\text{dB}$, and the data allocation probability $p$ to be $0.15$. As shown in Fig. \ref{fig:CNN-results}, we demonstrate the vulnerability of the naive Majority Vote-enabled SignSGD scheme in the presence of Byzantine attacks and the robustness of the proposed AirComp-enabled Hierarchical Vote framework to Byzantine attacks. Besides, by comparing the learning performance of the Hierarchical Vote with/without noise, the impact of the receiver noise on the convergence behavior is illustrated. Specifically, when SNR is of a small value, the learning performance of Hierarchical Vote degenerates compared with its noise-free version. Furthermore, we can conclude that the proposed AirComp-enabled Hierarchical Vote scheme can achieve satisfactory convergence results under both the convex and non-convex settings.

\section{Conclusions} \label{SecV}
In this paper, we developed an AirComp-enabled Hierarchical Vote framework in the wireless data center networks by exploiting the waveform superposition property of a multiple-access channel, which achieves communication-efficient design while maintaining convergence rate in the presence of Byzantine attacks. Through theoretical analysis, we demonstrated the effectiveness of the AirComp-enabled Hierarchical Vote framework in achieving robustness when tackling Byzantine attacks and receiver noise. We further characterized the convergence behavior of the proposed framework under non-convex settings and revealed the influence of Byzantine attacks and wireless environment on the convergence rate. The numerical simulation results demonstrated that the AirComp-enabled Hierarchical Vote scheme can achieve satisfactory learning performance in the presence of receiver noise and multiple types of Byzantine attacks.

\appendices
\section{Proof of Proposition \ref{Proposition-1}} \label{App-1}
For any edge worker $k$, given the index set of the allocated datasets denoted by $\mathcal{S}_k=\{k,i_1,\ldots,i_{n_k}\}$, we focus on the error probability of the $j$-th entry in $\pmb g_k^{(t)}$ while the others can be obtained via the same procedures. We define a Bernoulli random variable $X_{k,i},\,i\in\mathcal{S}_k$ to represent the relationship between the quantified local gradient $\pmb m_{k,i}^{(t)}(j)$ and the quantified true global gradient $\text{sign}\big(\pmb g^{(t)}\big)$, i.e., $X_{k,i}=\mathds{1}_{\pmb m_{k,i}^{(t)}(j)=\text{sign}\left(\pmb g^{(t)}(j)\right)}$. As revealed in \cite{SignSGD1,SignSGD2}, we have the following lemma:
\begin{lemma} \label{Lemma-1}
According to Assumption \ref{Assumption-4} and Definition \ref{Definition-2}, the error probability of the element-wise quantified local gradient can be upper-bounded by
\begin{equation} \label{Lemma-1-relationship}
q_{k,i}=\Pr\left(X_{k,i}=0\right)\leq\left\{
\begin{aligned}
&\frac{2}{9}\frac{1}{J_j^2}, &J_j\geq\frac{2}{\sqrt{3}}, \\
&\frac{1}{2}-\frac{J_j}{2\sqrt{3}}, &\text{otherwise},
\end{aligned}
\right.
\end{equation}
\end{lemma}
\noindent
We note that \emph{Lemma \ref{Lemma-1}} ensures the fact that the element-wise error probability $q_{k,i}$ from dataset $\mathcal{D}_{i}$ is less than or equal to $\frac{1}{2}$ in all cases. Recall the local Majority Vote scheme which selects the majority output from each allocated dataset in $\mathcal{S}_k$, we define a binomial random variable $X_k=\sum_{i\in\mathcal{S}_k}X_{k,i}$ to measure the voting process with
\begin{equation} 
\mathbb{E}(X_k)=(1-q_{k,i})|\mathcal{S}_k|,\quad\text{Var}(X_k)=q_{k,i}(1-q_{k,i})|\mathcal{S}_k|\label{EVXk-proof}
\end{equation}
Hence, the element-wise error probability of edge worker $k$ estimating $\pmb g^{(t)}$ can be given by
\begin{equation} 
\Pr\left(\pmb m_k^{(t)}(j)\neq\text{sign}\left(\pmb g^{(t)}(j)\right)\,\big|\,\mathcal{S}_k\right)=\Pr\left(X_k\leq\frac{|\mathcal{S}_k|}{2}\,\big|\,\mathcal{S}_k\right). \nonumber
\end{equation}
Note that the probability in this section is conditioned on the given $\mathcal{S}_k$, and we omit $\mathcal{S}_k$ in the sequel for notational ease. According to Cantelli's inequality, i.e., $\Pr(\mathbb{E}(X)-X\geq\lambda)\leq\frac{\text{Var}(X)}{\text{Var}(X)+\lambda^2},\,\lambda>0$, we can obtain
\begin{equation} 
\begin{split}
\Pr\left(X_k\leq\frac{|\mathcal{S}_k|}{2}\right)&=\Pr\left(\mathbb{E}(X_k)-X_k\geq\mathbb{E}(X_k)-\frac{|\mathcal{S}_k|}{2}\right) \\
&\leq\frac{\text{Var}(X_k)}{\text{Var}(X_k)+\left(\mathbb{E}(X_k)-\frac{|\mathcal{S}_k|}{2}\right)^2} \\ 
&\overset{(a)}{\leq}\frac{1}{2}\times\sqrt{\frac{\text{Var}(X_k)}{\left(\mathbb{E}(X_k)-|\mathcal{S}_k|/2\right)^2}}, \label{Xk-bound-proof}
\end{split}
\end{equation}
where (a) comes from the fact $1+a^2\geq2a$. By substituting (\ref{EVXk-proof}) into (\ref{Xk-bound-proof}), we can obtain
\begin{equation} 
\Pr\left(X_k\leq\frac{|\mathcal{S}_k|}{2}\right)\leq\frac{1}{2\sqrt{|\mathcal{S}_k|}}\times\sqrt{\frac{1}{4\left(\frac{1}{2}-q_{k,i}\right)^2}-1}. \label{Xk-bound1-proof}
\end{equation}
According to \emph{Lemma \ref{Lemma-1}}, (\ref{Lemma-1-relationship}) indicates the following upper bound
\begin{equation}
\frac{1}{4(\frac{1}{2}-q_{k,i})^2}-1\leq\frac{4}{J_j^2}, \label{Lemma-1-inference}
\end{equation}
whose proof can be easily found in the current references \cite{SignSGD1,Election-coding}. Finally, the expected result can be obtained by substituting (\ref{Lemma-1-inference}) into (\ref{Xk-bound1-proof}). \qed

\section{Proof of Theorem \ref{Theorem-1}} \label{App-2}
Note that $|\mathcal{S}_k|=n_k+1$ and $n_k$ is a random variable satisfying $n_k\sim\text{Binomial}(n,p)$. We can obtain the estimation error probability by averaging $\pmb q_{\mathcal{S}_k}^{(t)}(j)$ in \emph{Proposition \ref{Proposition-1}} over all the realization of $n_k$. Therefore, for a given coordinate $j$, the averaged bit error probability $\pmb q^{(t)}(j)$ can be upper-bounded by
\begin{equation} 
\begin{split}
\pmb q^{(t)}(j)&=\Pr\left(\pmb g_k^{(t)}(j)\neq\pmb g^{(t)}(j)\right)=\sum_{n_k=0}^{K-1}\Pr(n_k)\,\pmb q_{\mathcal{S}_k}^{(t)}(j) \\
&\leq\sum_{n_k=0}^{K-1}\frac{1}{J_j}\frac{1}{\sqrt{n_k+1}}\binom{K-1}{n_k}p^{n_k}(1-p)^{K-1-n_k} \\
&=\frac{1}{J_j}\mathbb{E}_{n_k}\left(\sqrt{\frac{1}{n_k+1}}\right).
\end{split}
\end{equation}
According to Jenson's inequality, i.e., $\mathbb{E}\left(f(x)\right)\leq f\left(\mathbb{E}(x)\right)$, and the concavity of the function $f(x)=\sqrt{x}$, we can obtain
\begin{equation} 
\frac{1}{J_j}\mathbb{E}_{n_k}\left(\sqrt{\frac{1}{n_k+1}}\right)\leq\frac{1}{J_j}\sqrt{\mathbb{E}_{n_k}\left(\frac{1}{n_k+1}\right)}.
\end{equation}
Based on the identity $\frac{1}{n_k+1}\binom{K-1}{n_k}=\frac{1}{K}\,\binom{K}{n_k+1}$, we can obtain
\begin{equation} 
\mathbb{E}_{n_k}\left(\frac{1}{n_k+1}\right)=\frac{1}{Kp}\left[1-(1-p)^K\right]\leq\frac{1}{Kp}. \nonumber
\end{equation}
Therefore, we arrive at the expected result, i.e.,
\begin{equation} 
\Pr\big[\pmb g_k^{(t)}(j)\neq\pmb g^{(t)}(j)\big]\leq\frac{1}{\,J_j\sqrt{Kp}\,}.
\end{equation}
Besides, to ensure that the honest edge workers can estimate the correct sign of the global model update with high probability, we restrict the error probability bound to less than $1/2$, i.e., $\frac{1}{J_j\sqrt{Kp}}<\frac{1}{2}$, yielding $p\in(\frac{4}{J_j^2 K},1]$, and we complete the proof. \qed

\section{Proof of Theorem \ref{Theorem-2}} \label{App-3}
The proving process is based on the procedures in Appendix \ref{App-1}. The key idea of this proof is to describe the event $\hat{\pmb v}^{(t)}=\text{sign}\left(\pmb g^{(t)}(j)\right)$ in one learning round. To this end, we define the following random variables, which are given by
\begin{equation} 
X_k=\left\{
\begin{aligned}
&\mathds{1}_{\pmb m_k^{(t)}(j)=\text{sign}\left(\pmb g^{(t)}(j)\right)},&k\in\mathcal{H}, \\
&0,&k\in\mathcal{B},
\end{aligned}
\right.
\end{equation}
to represent the relationship between the output from edge worker $k$ and the true gradient sign. Note that the honest edge workers transmit the correct message with probability $\pmb q^{(t)}(j)$ defined in \emph{Theorem \ref{Theorem-2}}, i.e., for any edge worker $k\in\mathcal{H}$, we have
\begin{equation} 
\begin{split}
\Pr\left(X_k=0\right)=\pmb q^{(t)}(j),\quad\Pr\left(X_k=1\right)=1-\pmb q^{(t)}(j),
\end{split}
\end{equation}
and the Byzantine edge workers transmit fake updates with probability $1$. We define $\widetilde{X}=\sum_{k\in\mathcal{K}}X_k+\tilde{n}$ to measure the over-the-air Majority Vote process, where $\tilde{n}$ is a Gaussian random variable satisfying
\begin{equation} 
\tilde{n}=\frac{1}{\rho^{(t)}} \text{Re}\left\{\pmb n^{(t)}(j)\right\}\sim\mathcal{N}\left(0,\frac{N_0}{2(\rho^{(t)})^2}\right). \nonumber
\end{equation}
We note that $\rho^{(t)}$ remains constant during one learning round, thereby this proving procedure can be extended to any coordinate $j\in[d]$. Recall that there exist $(1-c)K$ honest edge workers and $cK$ Byzantine edge workers, we derive the mean and variance of $\widetilde{X}$ as follows:
\begin{subequations} 
\begin{align}
&\mathbb{E}\left(\widetilde{X}\right)=(1-c)K\left(1-\pmb q^{(t)}(j)\right), \\
&\text{Var}\left(\widetilde{X}\right)=(1-c)K\pmb q^{(t)}(j)\left(1-\pmb q^{(t)}(j)\right)+\frac{N_0}{2(\rho^{(t)})^2}.
\end{align}
\end{subequations}
Note that we consider the worst-case scenario where the Byzantine edge workers and the receiver noise only corrupt the learning process, i.e., the receiver noise cannot reverse the sign of $\pmb m_k^{(t)}$ from Byzantine edge workers. Therefore, based on the Majority Vote scheme which outputs the majority opinion, we can obtain the necessary condition for a true output $\mathbb{E}(\widetilde{X})>\frac{K}{2}$, i.e.,
\begin{equation} 
(1-c)\left(1-\pmb q^{(t)}(j)\right)>\frac{1}{2},\;\;\forall j. \label{Necessary-condition-proof}
\end{equation}
For notational ease, we suppose $\zeta_j=(1-c)\left(1-\pmb q^{(t)}(j)\right)>1/2$.

\par
Further, according to (\ref{Aggregated-signal}) and (\ref{Central-majority-voting}), to ensure that
\begin{equation} 
\begin{split}
\text{sign}\big(\tilde{\pmb r}^{(t)}&(j)\big)=\text{sign}\left(\pmb g^{(t)}(j)\right) \\
&=\text{sign}\left(\sum_{k\in\mathcal{H}}\pmb m_k^{(t)}(j)+\sum_{k\in\mathcal{B}}\pmb m_k^{(t)}(j)+\frac{1}{\rho^{(t)}}\text{Re}\left\{\pmb n^{(t)}(j)\right\}\right), \nonumber
\end{split}
\end{equation}
$\widetilde{X}$ must be larger than $K/2$. Based on this, we can obtain the decoding bit error probability:
\begin{equation} 
P_{\text{err},j}^{(t)}=\Pr\left(\text{sign}\big(\tilde{\pmb r}^{(t)}(j)\big)\neq\text{sign}\left(\pmb g^{(t)}(j)\right)\right)=\Pr\left(\widetilde{X}\leq\frac{K}{2}\right) \nonumber
\end{equation}
According to Cantelli's inequality, i.e., $\Pr(\mathbb{E}(X)-X\geq\lambda)\leq\frac{\text{Var}(X)}{\text{Var}(X)+\lambda^2},\,\lambda>0$, we can obtain:
\begin{align*}
P_{\text{err},j}^{(t)}&\leq\frac{1}{2}\times\sqrt{\frac{\text{Var}(\widetilde{X})}{\left(\mathbb{E}(\widetilde{X})-K/2\right)^2}} \\
&=\frac{1}{2}\sqrt{\frac{(1-c)\pmb q^{(t)}(j)\left(1-\pmb q^{(t)}(j)\right)}{K\zeta_j^2}+\frac{N_0}{2\zeta_j^2\left(\rho^{(t)}\right)^2K^2}} \\
&\overset{(a)}{\leq}\frac{1}{2\zeta_j}\left(\sqrt{\frac{1-c}{K}\pmb q^{(t)}(j)\left(1-\pmb q^{(t)}(j)\right)}+\frac{1}{K\rho^{(t)}}\sqrt{\frac{N_0}{2}}\right), \nonumber
\end{align*}
where (a) comes from the fact $1+a^2\geq 2a$. We note that $\pmb q^{(t)}(j)\in(0,1)$, thus we can obtain $\pmb q^{(t)}(j)\left(1-\pmb q^{(t)}(j)\right)\leq 1/4$. Therefore, according to (\ref{Necessary-condition-proof}), we can obtain the decoding error probability for coordinate $j$, which is given by
\begin{equation} 
P_{\text{err},j}^{(t)}\leq\frac{1}{2}\sqrt{\frac{1-c}{K}}+\frac{1}{K\rho^{(t)}}\sqrt{\frac{N_0}{2}}.\label{Coordinate-bit-error-bound-proof}
\end{equation}
We note that this upper bound is not correlated to coordinate $j$ because the influence of $j$ is completely attributed to (\ref{Necessary-condition-proof}). Therefore, the decoding error probability $P_{\text{err}}^{(t)}$ is equal to $P_{\text{err},j}^{(t)}$ in (\ref{Coordinate-bit-error-bound-proof}) and we complete the proof. \qed

\section{Proof of Theorem \ref{Theorem-3}} \label{App-4}
According to Assumptions \ref{Assumption-1} and \ref{Assumption-2}, we first focus on the single-step loss bound:
\begin{equation}
\begin{split}
F(&\pmb\omega^{(t+1)})-F(\pmb\omega^{(t)}) \\
&\leq(\pmb g^{(t)})^T(\pmb\omega^{(t+1)}-\pmb\omega^{(t)})+\frac{1}{2}\sum_{j=1}^{d}L_j\left(\pmb\omega^{(t+1)}(j)-\pmb\omega^{(t)}(j)\right)^2 \\
&=-\eta\,(\pmb g^{(t)})^T\text{sign}\left(\tilde{\pmb r}^{(t)}\right)+\frac{1}{2}\sum_{j=1}^{d}L_j\left[-\eta\,\text{sign}\left(\tilde{\pmb r}^{(t)}(j)\right)\right]^2 \\
&=-\eta\,(\pmb g^{(t)})^T\text{sign}\left(\tilde{\pmb r}^{(t)}\right)+\frac{\eta^2}{2}\sum_{j=1}^{d}L_j \\
&=-\eta\,\big|\big|\pmb g^{(t)}\big|\big|_1+\frac{\eta^2}{2}\sum_{j=1}^{d}L_j \\
&\quad+2\eta\sum_{j=1}^{d}\big|\pmb g^{(t)}(j)\big|\,\mathds{1}_{\text{sign}\big(\tilde{\pmb r}^{(t)}(j)\big)\neq\text{sign}\left(\pmb g^{(t)}(j)\right)}. \nonumber
\end{split}
\end{equation}
Recall the random event $\text{sign}\big(\tilde{\pmb r}^{(t)}(j)\big)\neq\text{sign}\left(\pmb g^{(t)}(j)\right)$ analyzed in Appendix \ref{App-3}, by taking the expectation on both sides of the inequality conditioned on the previous model parameter $\pmb\omega^{(t)}$, we can obtain:
\begin{equation} 
\begin{split}
\mathbb{E}\Big(F(\pmb\omega^{(t+1)})&-F(\pmb\omega^{(t)})\,\big|\,\pmb\omega^{(t)}\Big)\leq-\eta\,\big|\big|\pmb g^{(t)}\big|\big|_1+\frac{\eta^2}{2}\sum_{j=1}^{d}L_j \\
&+2\eta\sum_{j=1}^{d}\big|\pmb g^{(t)}(j)\big|\,\Pr\left(\text{sign}\big(\tilde{\pmb r}^{(t)}(j)\big)\neq\text{sign}\big(\pmb g^{(t)}(j)\big)\right). \nonumber
\end{split}
\end{equation}

\par
Under the condition of \emph{Theorem \ref{Theorem-2}}, by substituting (\ref{Theorem-2-relationship}) into the above inequality, we have
\begin{align*} 
\mathbb{E}\Big(F(&\pmb\omega^{(t+1)})-F(\pmb\omega^{(t)})\,\big|\,\pmb\omega^{(t)}\Big)\leq-\eta\,\big|\big|\pmb g^{(t)}\big|\big|_1+\frac{\eta^2}{2}\sum_{j=1}^{d}L_j \\
&\quad+2\eta\sum_{j=1}^{d}\big|\pmb g^{(t)}(j)\big|\,\left(\frac{1}{2}\sqrt{\frac{1-c}{K}}+\frac{1}{K\rho^{(t)}}\sqrt{\frac{N_0}{2}}\right) \\
&=\eta\left(\sqrt{\frac{1-c}{K}}+\frac{\sqrt{2}}{K\rho^{(t)}}\sqrt{N_0}-1\right)\big|\big|\pmb g^{(t)}\big|\big|_1+\frac{\eta^2}{2}||\pmb L||_1 \\
&\overset{(a)}{=}\frac{1}{\,\sqrt{T||\pmb L||_1}}\left(\sqrt{\frac{1-c}{K}}+\frac{\sqrt{2}}{K\rho^{(t)}}\sqrt{N_0}-1\right)\big|\big|\pmb g^{(t)}\big|\big|_1+\frac{1}{2T}, \nonumber
\end{align*}
where (a) comes from the setting $\eta=1\big/\sqrt{T||\pmb L||_1}$. We consider the worst-case scenario and define $\rho_{\text{min}}=\min_{t}\rho^{(t)}$, which refers to the lowest transmit power during the learning process.

\par
According to Assumption \ref{Assumption-2}, the global loss function defined in (\ref{Global-loss-function}) has a lower bound $F^*$. To average out the randomness, we take the expectation over $\pmb\omega^{(t)}$, which leads to
\begin{equation} 
\begin{split}
F&(\pmb\omega^{(0)})-F^*\geq F(\pmb\omega^{(0)})-\mathbb{E}\left(F(\pmb\omega^{(T)})\right) \\
&\geq\sum_{t=0}^{T-1}\mathbb{E}\left[F(\pmb\omega^{(t)})-F(\pmb\omega^{(t+1)})\right] \\
&\geq\sum_{t=0}^{T-1}\mathbb{E}\left[\frac{1}{\,\sqrt{T||\pmb L||_1}}\left(1-\sqrt{\frac{1-c}{K}}-\frac{\sqrt{2}}{K\rho_{\text{min}}}\sqrt{N_0}\right)\big|\big|\pmb g^{(t)}\big|\big|_1\!-\!\frac{1}{2T}\right] \\
&=\frac{1}{\,\sqrt{T||\pmb L||_1}}\left(1-\sqrt{\frac{1-c}{K}}-\frac{\sqrt{2}}{K\rho_{\text{min}}}\sqrt{N_0}\right)\mathbb{E}\left[\sum_{t=0}^{T-1}\big|\big|\pmb g^{(t)}\big|\big|_1\right]\!-\!\frac{1}{2T} \\
&\triangleq\frac{1}{\,\sqrt{T||\pmb L||_1}}\,\Delta\,\mathbb{E}\left[\sum_{t=0}^{T-1}\big|\big|\pmb g^{(t)}\big|\big|_1\right]-\frac{1}{2T} \nonumber
\end{split}
\end{equation}

\par
Finally, we rearrange the above formula and obtain
\begin{equation} 
\mathbb{E}\left[\sum_{t=0}^{T-1}\big|\big|\pmb g^{(t)}\big|\big|_1\right]\leq\frac{\sqrt{T||\pmb L||_1}}{\Delta}\left(F(\pmb\omega^{(0)})-F^*+\frac{1}{2T}\right).
\end{equation}
The expected result can be obtained by dividing both sides of the inequality by $T$. \qed

\bibliographystyle{ieeetr}
\bibliography{refs}

\begin{IEEEbiography}[{\includegraphics[width=1in,height=1.25in,clip,keepaspectratio]{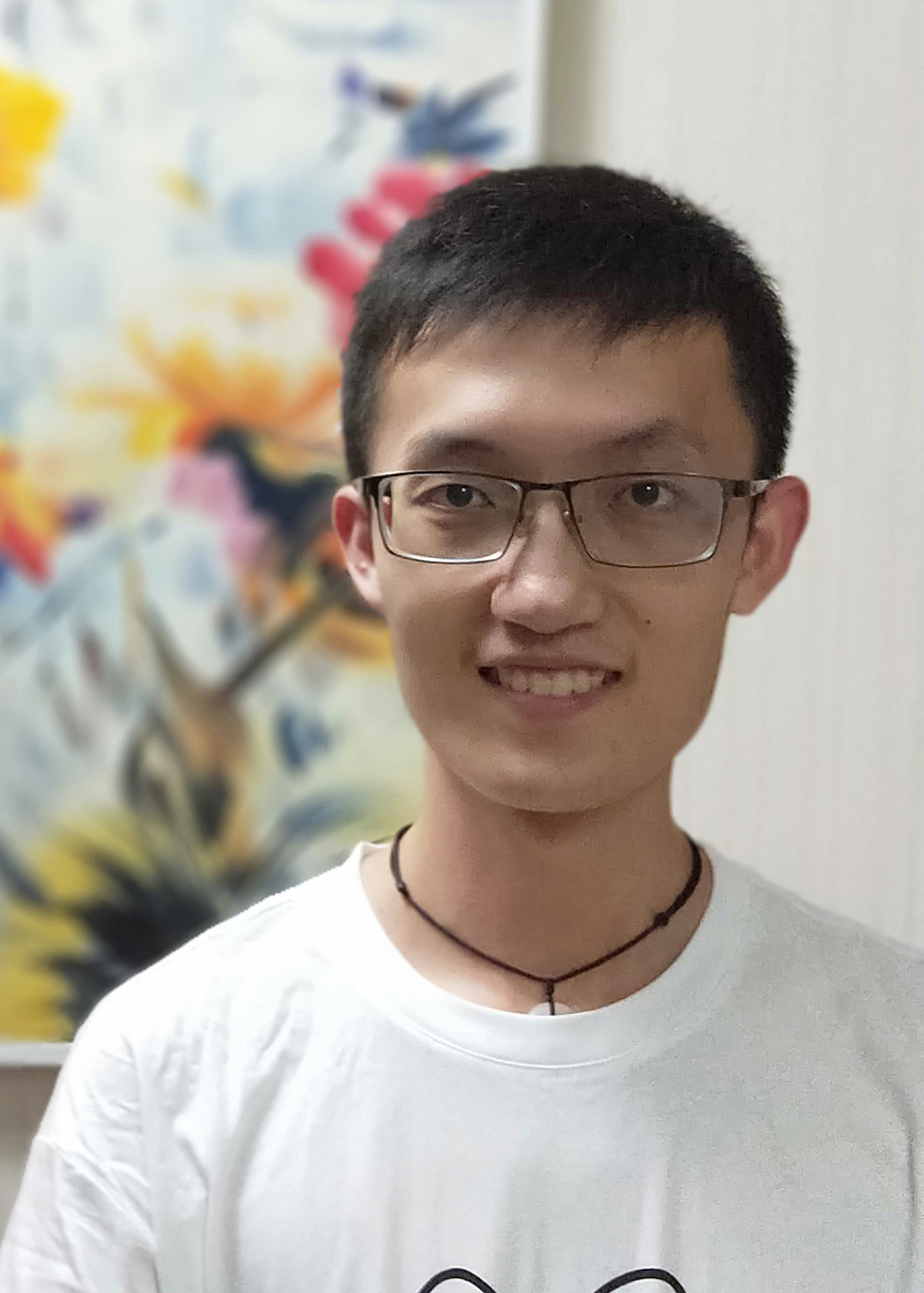}}]{Yuhan Yang} (Student Member) received the B.S. degree in communication engineering from Xidian University, Xi’an, China, in 2021. He is currently pursuing the master’s degree with the School of Information Science and Technology, ShanghaiTech University. His research interests include wireless communications, information theory and trustworthy distributed learning.
\end{IEEEbiography}
\vskip 0pt plus -1fil
\begin{IEEEbiography}[{\includegraphics[width=1in,height=1.25in,clip,keepaspectratio]{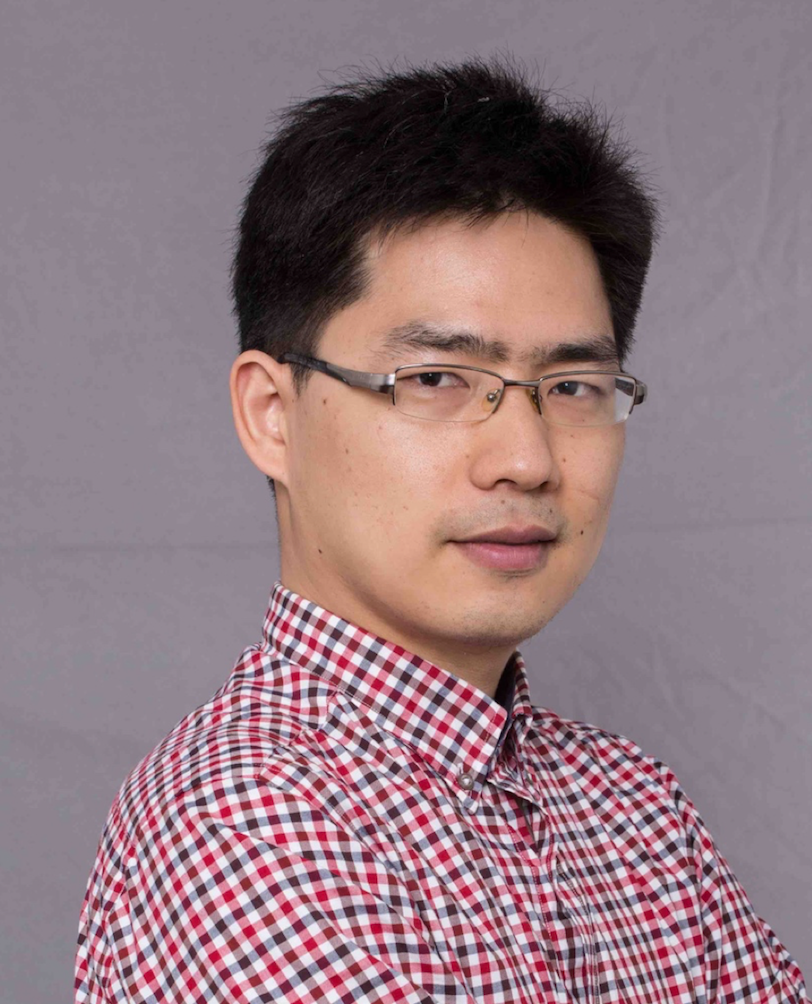}}]{Youlong Wu} (S’13–M’15) obtained his B.S. degree in electrical engineering from Wuhan University, Wuhan, China, in 2007. He received the M.S. degree in electrical engineering from Shanghai Jiaotong University, Shanghai, China, in 2011. In 2014, he received the Ph.D. degree at Telecom ParisTech, in Paris, France. In December 2014, he worked as a postdoc at the Institute for Communication Engineering, Technical University Munich (TUM), Munich, Germany. In 2017, he joined the School of Information Science and Technology at ShanghaiTech University. He obtained the TUM Fellowship in 2014 and is an Alexander von Humboldt research fellow. His research interests in Communication Theoy, Information Theory and its applications e.g., coded caching, distributed computation, and machine learning.
\end{IEEEbiography}
\vskip 0pt plus -1fil  
\begin{IEEEbiography}
[{\includegraphics[width=1in,height=1.25in,clip,keepaspectratio]{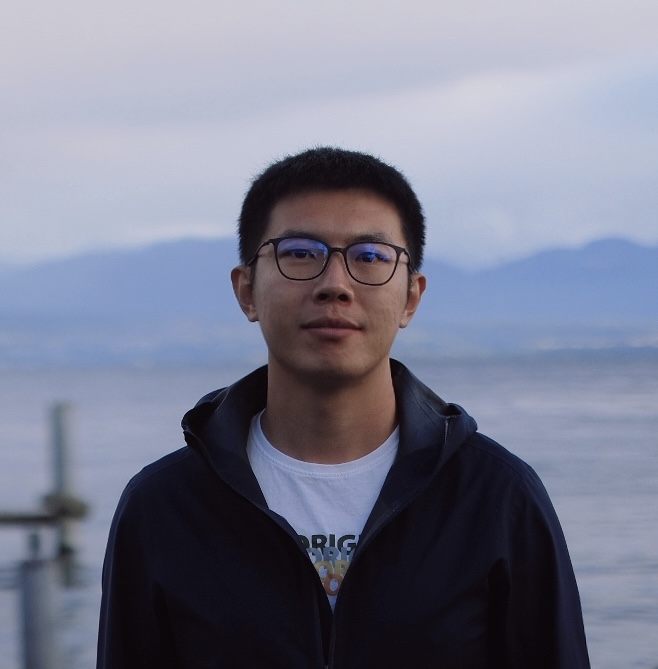}}]{Yuning Jiang} (Member, IEEE) received the B.Sc. degree in electronic engineering from Shandong University, Jinan, China, in 2014, and the Ph.D. degree in information engineering from ShanghaiTech University, Shanghai, China, and the University of Chinese Academy of Sciences, Beijing, China, in 2020. He was a Visiting Scholar with the University of California at Berkeley (UC Berkeley), Berkeley, CA, USA, the University of Freiburg, Freiburg im Breisgau, Germany, and Technische Universität Ilmenau (TU Ilmenau), Ilmenau, Germany, during his Ph.D. study. He is currently a Postdoctoral Researcher with the Automatic Control Laboratory, École Polytechnique Fédérale de Lausanne (EPFL), Lausanne, Switzerland. His research focuses on learning- and optimization-based policy for operating complex systems, such as nonlinear autonomous systems (e.g., autonomous vehicles, robotics, and smart buildings), and large-scale multiagent systems (e.g., power and energy systems, IoT, and traffic networks).
\end{IEEEbiography}
\vskip 0pt plus -1fil  
\begin{IEEEbiography}[{\includegraphics[width=1in,height=1.25in,clip,keepaspectratio]{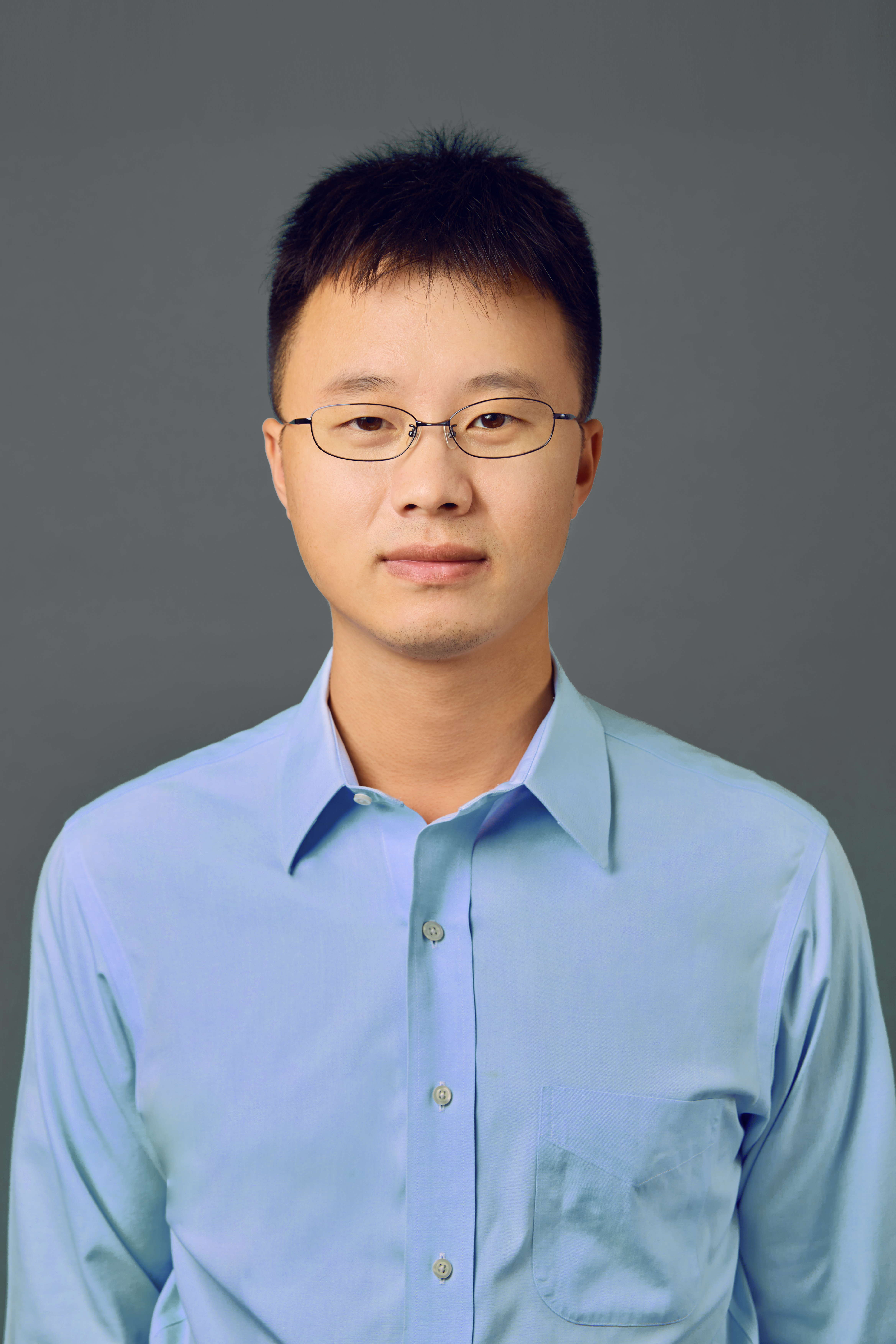}}]{Yuanming Shi} (S’13-M’15-SM’20) received the B.S. degree in electronic engineering from Tsinghua University, Beijing, China, in 2011. He received the Ph.D. degree in electronic and computer engineering from The Hong Kong University of Science and Technology (HKUST), in 2015. Since September 2015, he has been with the School of Information Science and Technology in ShanghaiTech University, where he is currently a tenured Associate Professor. He visited University of California, Berkeley, CA, USA, from October 2016 to February 2017. His research areas include edge AI, wireless communications, and satellite networks. He was a recipient of the IEEE Marconi Prize Paper Award in Wireless Communications in 2016, the Young Author Best Paper Award by the IEEE Signal Processing Society in 2016, the IEEE ComSoc Asia-Pacific Outstanding Young Researcher Award in 2021, and the Chinese Institute of Electronics First Prize in Natural Science in 2022. He is also an editor of IEEE Transactions on Wireless Communications, IEEE Journal on Selected Areas in Communications, and Journal of Communications and Information Networks. He is an IET Fellow.
\end{IEEEbiography}
\end{document}